\documentclass[11pt]{article}
\usepackage[utf8]{inputenc}
\usepackage[T1]{fontenc}
\usepackage{xcolor}
\usepackage{amsmath,amsxtra,amssymb,latexsym,amscd,amsthm,pb-diagram,fancyhdr,euscript}
\usepackage{indentfirst}
\usepackage{epsfig}
\usepackage{mathrsfs}
\usepackage{slashed}
\usepackage{csquotes}
\usepackage{caption}
\usepackage{subcaption}
\usepackage{tikz}
\usetikzlibrary{decorations.markings,decorations.pathmorphing}
\usepackage{hyperref}
\hypersetup{
   colorlinks,
    citecolor=blue,
    filecolor=black,
    linkcolor=red,
    urlcolor=magenta
}
\hypersetup{linktocpage}

\newcommand{\R}{\mathbb{R}} 

\newcommand{\p}{\partial}
\renewcommand{\d}{\mathrm{d}}

\newcommand{\scri}{{\mathscr I}}
\newcommand{\scrh}{{\mathscr H}}

\DeclareMathOperator{\Int}{Int}
\DeclareMathOperator{\Ext}{Ext}

\newcommand{\M}{\mathcal{M}}

\newtheorem{theorem}{Theorem}[section]
\newtheorem{proposition}{Proposition}[section]
\newtheorem{corollary}{Corollary}[section]
\newtheorem{lemma}{Lemma}[section]
\newtheorem{remark}{Remark}[section]
\begin{document}

\numberwithin{equation}{section}
\begin{center}
\bf{\Huge Peeling at extreme black hole horizons}

\vspace{0.1in}

{Jack BORTHWICK\footnote{LmB, UMR CNRS 6623, Université de Bourgogne Franche-Comté, 16 route de Gray, 25030 Besançon cedex, France, jack.borthwick@univ-fcomte.fr}, Eric GOURGOULHON\footnote{LUTH, UMR CNRS 8102, Observatoire de Paris, Université PSL, Université Paris Cité, 5 place Jules Janssen, 92190 Meudon, France, eric.gourgoulhon@obspm.fr} \& Jean-Philippe NICOLAS\footnote{LMBA, UMR CNRS 6205, Université de Brest, 6 avenue Victor Le Gorgeu, 29200 Brest, France, jean-philippe.nicolas@univ-brest.fr} }
\end{center}

{\bf Abstract.}
The starting point of this work was an intriguing similarity between the behaviour of fields near a degenerate horizon and near the infinity of an asymptotically flat spacetime, as revealed by the scattering theory for Dirac fields in the \enquote{exterior} region of the extreme Kerr - de Sitter black hole, developed by one of the authors (JB). However, in that situation, the comparison was somewhat clouded by some of the analytical techniques used in intermediate steps of the proof. The aim of the present work is to clarify the comparison further by studying instead the peeling behaviour of solutions to the wave equation at an extremal horizon. We focus first on the extreme Reissner-Nordström black hole, for which the Couch-Torrence inversion (a global conformal isometry that exchanges the horizon and infinity) makes the analogy explicit. Then, we explore more general spherically symmetric situations using the Couch-Torrence inversion outside of its natural context.
\vspace{0.1in}


\vspace{0.1in}


\tableofcontents

\section{Introduction}

There is a similarity of behaviour for a field propagating near a degenerate horizon and near infinity. One of us (Jack Borthwick \cite{JB2020}) has recently studied the scattering of Dirac fields by an extremal Kerr-de Sitter black hole; the scattering theory is obtained via spectral methods using Mourre theory, inferring from it propagation estimates that serve as weak versions of the Huygens principle and allow to compare the physical dynamics with a succession of simplified dynamics. Near the horizon of the black hole, that is the degenerate one, the last and simplest comparison dynamics is the radial part of the Dirac Hamiltonian on Minkowski spacetime. This indicates an analogy, but the presence of the other comparison dynamics, one of which involving a Dollard-type phase modification, makes the link rather obscure. The main purpose of this paper is to clarify this analogy, from the point of view of the peeling behaviour instead of a scattering theory.

The peeling is a type of asymptotic behaviour at infinity along outgoing null geodesics, satisfied by zero rest-mass fields on Minkowski spacetime, that was discovered by Sachs in the early 1960's \cite{Sa61}. The original description by Sachs states that an outgoing zero rest-mass field of spin $s$, along a null geodesic going out to infinity, can be expanded in powers of $1/r$ and the part of the field falling-off like $r^{-k}$, $1\leq k \leq 2s$, has $2s-k$ of its principal null directions aligned along the null geodesic. In 1965, Roger Penrose \cite{Pe65} proved that this is equivalent to a much simpler property, the continuity at null infinity of the rescaled field. In 2009, Lionel Mason and one of the authors (Lionel Mason and Jean-Philippe Nicolas \cite{MaNi2009}) studied the peeling for scalar fields on the Schwarzschild spacetime. They redefined the notion of peeling using a characterisation of regularity at any order at null infinity in terms of Sobolev-type spaces, obtained as energy fluxes for a special observer: the ``Morawetz vector field''. This is in the same spirit as Penrose's version of the original definition, but is more amenable to analysis since such function spaces are naturally preserved under the evolution for hyperbolic equations, unlike $\mathcal{C}^k$ spaces. They provided a complete description of the classes of initial data on a Cauchy hypersurface that give rise to a peeling at any given order at future null infinity. Comparing the construction to the analogous one on Minkowski spacetime, it turned out that these classes are the same in both spacetimes, in terms of regularity and decay at infinity. This means that in spite of the different asymptotics of the two metrics, the conditions for peeling are not more stringent on the Schwarzschild metric than on Minkowski spacetime; they are identical. The original Morawetz vector field is one of the conformal Killing vectors of Minkowski spacetime. It was discovered by Kathleen Morawetz in 1962 \cite{Mo1962} and used to establish decay properties for solutions to the wave equation on flat spacetime. The vector field used in \cite{MaNi2009} is a natural adaptation of this vector field to the Schwarzschild metric. The first occurrence of such a modified (and slightly different) Morawetz vector field was in a paper by Inglese and Nicolò in 2000 \cite{InglNi}.

A typical example of degenerate horizon can be found in the extreme Reissner-Nordström spacetime. In this case, there exists a remarkable conformal isometry of the exterior of the black hole that exchanges the horizon and infinity. This was initially discovered by Couch and Torrence in 1983 \cite{CoTo1983} and is referred to as the Couch-Torrence inversion. A nice description of it with useful additional properties can be found in Aretakis's book \cite{Are2018Book} as well as in Bizon and Friedrich \cite{BiFri2013} and in Lübbe and Valiente-Kroon \cite{LVK2014}. This inversion makes the above-mentioned analogy of behaviour very precise on the extreme Reissner-Nordström geometry. First, extending the results of \cite{MaNi2009} to infinity of the extreme Reissner-Nordström spacetime, we then translate them at the extreme horizon via the Couch-Torrence inversion.  The peeling at the horizon is analogous to the peeling at the conformally rescaled infinity, modulo a finite conformal rescaling of the horizon. This entails that massive fields can also be shown to exhibit a complete peeling at the horizon.
Such precise global structures as the Couch-Torrence inversion do not seem to exist for other extremal black hole spacetimes in four dimensions. Nevertheless, one may ask if they can be localised near extremal horizons in order to provide an alternative description of their neighbourhoods in a useful way. We provide here a first example of this type of construction and use it to study the peeling near a class of spherically symmetric degenerate horizons.

Our paper is organised as follows. Section \ref{GB} presents the conformal d'Alembertian, the extreme Reissner-Nordström spacetime and its conformal compactification. In Section \ref{CTSection}, we describe the Couch-Torrence inversion on the extreme Reissner-Norström metric and observe something that seems to have been overlooked until now in the literature: there exist many choices of conformal rescalings of the extreme Reissner-Norström metric that make the Couch-Torrence inversion an isometry (and not just a conformal isometry), among which the simplest and most useful rescaling associated with the conformal factor $\Omega = 1/r$. Other examples are given as well as the condition on the conformal factor for this to be true. Section \ref{PeelingRNE} is devoted to the peeling at infinity on extreme Reissner-Nordström spacetime and its translation at the horizon, including for massive fields, using the Couch-Torrence inversion. Finally, in Section \ref{LooseEndCT}, we consider a class of spherically symmetric degenerate horizons and apply the Couch-Torrence inversion to describe them as conformally rescaled infinities. We then study the peeling of scalar fields at these infinities and translate back to the horizons.

\medskip
\noindent{\textbf{Notations}:} $ A \lesssim B$ means that there is a constant $C>0$ such that $A\leq C B$, $A \simeq B$ means that there are constants $c,C>0$ such that: $cB \leq A \leq CB$. Given a smooth differentiable manifold $\mathcal{M}$, we denote by $\mathcal{C}^\infty_0 (\mathcal{M})$ the space of smooth and compactly supported functions on $\mathcal{M}$. Throughout the paper, we use the abstract index formalism of Penrose and Rindler \cite{PeRi1984}.

\medskip
\noindent{\textbf{Notebooks}:} Some computations performed in this article are detailed in
the following publicly available SageMath \cite{Sage} notebooks:
\begin{itemize}
\item Sections~\ref{GB} to \ref{PeelingRNE}:\\
{\small \url{https://cocalc.com/share/public_paths/f05d583cb13735d5f9e8bd292d34b24633738129}}
\item Section~\ref{LooseEndCT}:\\
{\small \url{https://cocalc.com/share/public_paths/795aec91a03a81ed3ce0154e6fe9bafb7b69845c}}
\end{itemize}

\section{Geometrical background} \label{GB}

\subsection{The conformal d'Alembertian and wave equation}\label{IntroConformaldAlembertian}
Let $(\M,\mathbf{c})$ be a conformal manifold of dimension $n$ and suppose that $\mathbf{c}=[g]$ is the conformal class of a Lorentzian metric $g$. Call $\mathcal{E}[\omega]$ the module of conformal densities (see \S 2.4 in \cite{CuGo2015}) on $\M$ of weight $\omega \in \mathbb{R}$.  Conformal densities of weight  $-n$ can be identified, given a choice of $g \in \mathbf{c}$, with usual $1$-densities via the map:
\[ \sigma \mapsto \frac{\sigma}{\textrm{Vol}_g}, \]
where $\textrm{Vol}_g$ is the canonical volume density of $g$.
For every $g \in \mathbf{c}$, there is a canonical conformal density\footnote{roughly the $\frac{1}{-n}$th power of the image of $\textrm{Vol}_g$ under the above map} of weight $1$ that we denote by $\sigma_g$ and that is parallel for the Levi-Civita connection of $g$. Observe furthermore that $\mathbf{g}=g\sigma_g^2$ is conformally invariant; it is referred to as the conformal metric (it has conformal weight $2$).
It is well known that the operator, expressed in terms of an arbitrary metric $g\in\mathbf{c}$:
\[ \Box_\mathbf{c} = \mathbf{g}^{ab}\nabla_a\nabla_b + \frac{1}{6}\mathbf{g}^{ab}R_{ab},\]
is conformally invariant acting from $\mathcal{E}[-1]$ into $\mathcal{E}[-3]$.
In the above equation, $\nabla$ and $R_{ab}$ are the Levi-Civita connection and Ricci tensor $R_{ab}=R_{ac\phantom{c}b}^{\phantom{ac}c}$ \emph{of the chosen metric $g$} respectively and conformal invariance is to be understood to mean that the result of the above computation \emph{does not depend} on the choice of $g\in \mathbf{c}$ used to calculate it. We will refer to it as the \emph{conformal d'Alembertian} and the equation:
\begin{equation} \label{CWEDensity}\Box_{\mathbf{c}} \underline{\phi} =0, \qquad \underline{\phi} \in \mathcal{E}[-1],\end{equation}
as the \emph{conformal wave equation}.

In practice, this shall be exploited as follows. Let $g\in \mathbf{c}$ be given and suppose that $\phi$ is a scalar field that satisfies: \begin{equation} \label{CWaveScalar}\Box_g \phi +\frac{1}{6}g^{ab}R_{ab}\phi =0,\end{equation} with $\square_g = g^{ab} \nabla_a \nabla_b$, then the conformal density $\underline{\phi}=\phi \sigma_g^{-1}$ satisfies $\Box_\mathbf{c} \underline{\phi}=0$. Introducing, $\hat{g}=\Omega^2g$, we see that $\underline{\phi}=\Omega^{-1} \phi  \sigma_{\hat{g}}^{-1}$ hence, $\hat{\phi}=\Omega^{-1}\phi$ satisfies: \[\Box_{\hat{g}}\hat{\phi}+\frac{1}{6}\hat{g}^{ab}\hat{R}_{ab}\hat{\phi}=0.\]

\noindent For us $g$ will be the physical metric, and $\hat{g}$ the compactified metric.

\subsection{Extremal Reissner-Nordström metric and its conformal compactification}
\label{s:ERN_conformal_compact}
Let us recall some basic facts about the extremal Reissner-Nordström metric. The standard expression for the metric in Schwarzschild-like coordinates $(t,r,\omega)\in \R\times \R_+^*\times S^2$ is
\[ g = F(r) \d t^2 - \frac{1}{F(r)} \d r^2 - r^2 \d \omega^2 \, ,~ F(r) = \left( \frac{r-M}{r} \right)^2 \, ,\]
$\d \omega^2$ being the standard round metric on the unit $2$-sphere. Of course, this is only defined separately on the open sets $\Int=\R_t\times (0,M)_r\times S^2$ and $\Ext=\R_t\times(M,+\infty)_r\times S^2$.

Introduce now Regge-Wheeler's tortoise type coordinate $r_*$, defined by:
\begin{align}
\frac{\d r_*}{\d r} = \frac{1}{F(r)}&=1+\frac{M^2}{(r-M)^2}+\frac{2M}{r-M} \,\label{EDOrstar} \\
r_* = 0 &\Leftrightarrow r=2M \, , \label{Constantrstar}
\end{align}
which is easily integrated to yield:
\begin{equation} \label{rstar}
r_* = r -M + 2M \log \left( \frac{r-M}{M} \right) - \frac{M^2}{r-M} \, .
\end{equation}
We have chosen $r_*$ so that it vanishes on the photon sphere, located at $r=2M$.

The metric $g$ can be extended analytically across the coordinate singularity $\{r=M\}$ in two different ways:
\begin{enumerate}
\item In outgoing Eddington-Finkelstein coordinates, defined by:
\begin{equation}\label{ef-outgoing} u=t-r_*, r, \omega, \end{equation}
leading to $E_+= \R_u\times \R_+^* \times S^2$, endowed with the metric $g_1$:
\[g_1= F(r)\d u ^2 +2\d u \d r -r^2 \d \omega^2, \]
see Figure~\ref{outgoing}. The past event horizon $\scrh^-$ is
the hypersurface $r=M$ of $E_+$.
\item In ingoing Eddington-Finkelstein coordinates, defined by:
\begin{equation}\label{ef-ingoing} v= t+r_*,r,\omega, \end{equation}
leading to $E_-= \R_v \times \R_+^* \times S^2$, equipped with:
\[g_2 = F(r)\d v^2 -2\d v\d r -r^2 \d \omega^2, \]
see Figure~\ref{ingoing}. The future event horizon $\scrh^+$ is
the hypersurface $r=M$ of $E_-$.
\end{enumerate}
\begin{figure}[h!]
\centering
\begin{subfigure}{0.3\paperwidth}
\begin{tikzpicture}[scale=1.25,very thick]
\draw [](-1,-1) --node[below]{$\scrh^+$} (1,-1);
\draw [, decorate,decoration={zigzag,amplitude=2,segment length=3mm}](-1,-1) --node[above,rotate=45]{$r=0$}  (1,1);
\draw [](1,1)--node[]{$\scrh^{-}$}(1,-1);
\draw [](1,1)--node[above]{$\scrh^+$}(3,1);
\draw [](1,-1)--node[below]{$\scri^-$}(3,-1);
\draw [](3,1)--node[right]{$\scri^+$}(3,-1);
\draw[,postaction={decorate,decoration={markings,mark=at position .31 with {\arrow{latex}},mark=at position .9 with {\arrow{latex}}}}](-.5,-.5)--(3,-.5);
\draw[,postaction={decorate,decoration={markings, mark=at position 0.75 with {\arrow{latex}}}}](.2,-1)--(.2 ,.2);
\draw[,postaction={decorate,decoration={markings,mark=at position .75 with {\arrow{latex}}}}](1.5,-1)--(1.5,1);
\node at (2,.5){$\Ext$};
\node at (.5,.1){$\Int$};
\end{tikzpicture}
\caption{\label{outgoing}$E_+$}\end{subfigure}
\begin{subfigure}{0.3\paperwidth}
\begin{tikzpicture}[scale=1.25,very thick]
\draw [](-1,1) --node[above]{$\scrh^-$}(1,1);
\draw [, decorate,decoration={zigzag,amplitude=2,segment length=3mm}](-1,1) --node[below,rotate=-45]{$r=0$}  (1,-1);
\draw [](1,1)--node[below]{$\scrh^+$}(1,-1);
\draw [](1,1)--node[above]{$\scri^+$}(3,1);
\draw [](1,-1)--node[below]{$\scrh^-$}(3,-1);
\draw [](3,1)--node[right]{$\scri^-$}(3,-1);
\draw[,postaction={decorate,decoration={markings,mark=at position 0.3 with {\arrow{latex}},mark=at position .9 with {\arrow{latex}}}}](3,.23)--(-.23,.23);
\draw[,postaction={decorate,decoration={markings, mark=at position 0.85 with {\arrow{latex}}}}](2.5,-1)--(2.5,1);
\draw[,postaction={decorate,decoration={markings,mark=at position .8 with {\arrow{latex}}}}](0.5,-.5)--(0.5,1);
\node at (1.75,.6){$\Ext$};
\node at (0,.6){$\Int$};
\end{tikzpicture}
\caption{\label{ingoing}$E_-$}\end{subfigure}
\caption{Schematic representation of $E_+$ and $E_-$, the horizontal and vertical lines in each block represent future-oriented principal null geodesics. The terminology future/past horizon is to be understood with respect to the exterior block.}
\end{figure}
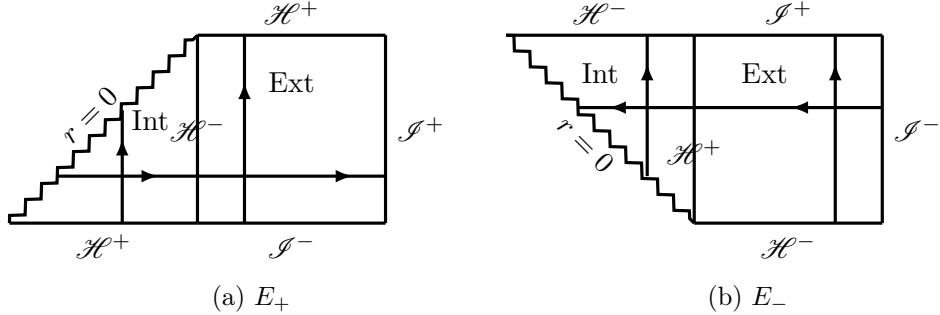
The time-orientation of both spaces is determined by $\frac{\partial}{\partial t}$ for large $r$.
The construction is based on the two distinguished null directions -- the principal null directions -- defined in Schwarzschild coordinates by
\begin{equation}
n_+ = \frac{1}{F(r)}\frac{\partial }{\partial t}+ \frac{\partial}{\partial r}, \quad n_-=\frac{1}{F(r)}\frac{\partial}{\partial t} -\frac{\partial}{\partial r}.
\end{equation}
In both $E_+$ and $E_-$, $\Int$ and $\Ext$ are easily identified as open sets and $\Int \cup \Ext$ is dense. In $E_+$, the coordinate lines of $r$ are  the integral curves of $n_+$ and are (complete) geodesics; these are referred to as the outgoing principal null geodesics. In $E_+$ ingoing principal null geodesics are defined as geodesic reparametrisations of $F(r)n_- $; they are however (future) incomplete. The situation is reversed in $E_-$: the coordinate lines of $r$ (oriented for decreasing $r$) are integral curves of $n_-$ and are complete geodesics referred to as the ingoing principal null geodesics. Outgoing principal null geodesics are the (past)-incomplete geodesic reparametrisations of $F(r)n_+$.
Gluing $E_+$ and $E_-$ together according to the tiling represented in Figure~\ref{tiling} leads to a maximal analytical extension of the extreme Reissner-Nordström metric. In any of the $\Ext$ blocks there is a past and future horizon, however they are not joined by a crossing sphere (since the geodesics on the horizons are complete); there is instead an \enquote{internal infinity}, which we shall denote by $i^1$.
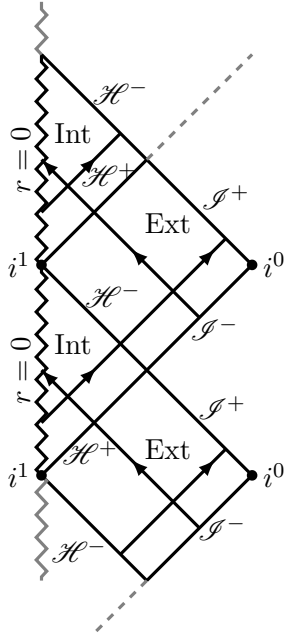
\begin{figure}
\centering
\begin{tikzpicture}[scale=1.4, very thick]
\draw[decorate,decoration={zigzag,amplitude=2,segment length= 3mm}](0,-1) --node[above,rotate=90]{$r=0$} (0,1);
\draw (0,-1) --node[xshift=-.5mm,yshift=-4mm]{$\scrh^+$} (1,0);
\draw (1,0) --node[xshift=2.5mm,yshift=3mm]{$\scrh^-$} (0,1);
\draw[postaction={decorate,decoration={markings,mark=at position .7 with \arrow{latex}}}] (0,-.5)--(0.75,.25);
\draw[postaction={decorate,decoration={markings,mark=at position 1 with \arrow{latex}}}](0.5,-.5)--(0,0);
\node at (.3,0.25){$\Int$};
\draw (0,-1) --node[xshift=-2mm,yshift=-3mm]{$\scrh^-$} (1,-2);
\draw (1,0) --node[xshift=3mm,yshift=2mm]{$\scri^+$} (2,-1);
\filldraw[black] (2,-1) circle (1pt) node[anchor=west]{$i^0$};
\draw (2,-1) --node[xshift=3mm]{$\scri^-$} (1,-2);
\draw[postaction={decorate,decoration={markings,mark=at position .65 with \arrow{latex}}}](1.5,-1.5)--(0.5,-.5);
\draw[postaction={decorate,decoration={markings,mark=at position .9 with \arrow{latex}}}] (.75,-1.75)--(1.75,-.75);
\node at (1.2,-.75){$\Ext$};
\draw[gray,decorate,decoration={zigzag,amplitude=2,segment length=3mm}] (0,-1) -- (0,-2);
\draw[gray,dashed] (1,-2) -- (.5,-2.5);
\filldraw[black] (0,-1) circle (1pt) node[anchor=east]{$i^1$};
\draw (0,1) --node[xshift=2mm,yshift=5mm]{$\scrh^+$}(1,2);
\draw (1,0)--node[xshift=2mm, yshift=-1.5mm]{$\scri^-$} (2,1);
\filldraw[black] (2,1) circle (1pt) node[anchor=west]{$i^0$};
\draw (2,1)--node[xshift=3.5mm,yshift=2mm]{$\scri^+$}(1,2);
\draw[postaction={decorate,decoration={markings,mark=at position .9 with \arrow{latex}}}] (0.75,.25)--(1.75,1.25);
\draw[postaction={decorate,decoration={markings,mark=at position .65 with \arrow{latex}}}](1.5,.5)--(.5,1.5);
\filldraw[black] (0,1) circle (1pt) node[anchor=east]{$i^1$} ;
\node at (1.2,1.4){$\Ext$};

\draw[decorate,decoration={zigzag,amplitude=2,segment length=3mm}] (0,1) --node[above,rotate=90]{$r=0$} (0,3);
\draw[postaction={decorate,decoration={markings,mark=at position 1 with \arrow{latex}}}](.5,1.5)--(0,2);
\draw[postaction={decorate,decoration={markings,mark=at position .7 with \arrow{latex}}}] (0,1.5)--(.75,2.25);
\draw (1,2) --node[xshift=3.5mm,yshift=2mm]{$\scrh^ -$}(0,3);
\draw[gray,decorate,decoration={zigzag,amplitude=2,segment length=3mm}] (0,3) -- (0,3.5);
\draw[gray,dashed] (1,2) -- (2,3);
\node at (.3,2.25){$\Int$};
\end{tikzpicture}
\caption{\label{tiling}Tiling constructed from $E_+$ and $E_-$ that completes the ingoing and outgoing principal null geodesics.
} 
\end{figure}

A set of variables that will be useful to us for studying peeling properties are the outgoing Eddington-Finkelstein coordinates with an inversion in $r$
\begin{equation}
R = \frac{1}{r} \, ,~ u = t-r_* \, ,~ \omega \, .
\end{equation}
These are also adapted to the $1/r$-compactification of the exterior region equipped with the metric $\hat{g}$ given by
\begin{equation} \label{CompExtRNMetEF}
\hat{g} = R^2 g = R^2 (1-MR)^2 \d u^2 -2 \d u \d R - \d \omega^2 \, .
\end{equation}
The future null infinity $\scri^+$ is then the hypersurface $\{R=0\}$.
Note that it is a degenerate Killing horizon with respect to the Killing
vector $\partial/\partial u$ of $\hat{g}$. Similarly, the ingoing Eddington-Finkelstein coordinates with an inversion in $r$:
\begin{equation}
R=\frac1{r} \, ,~v=t+r_* \, ,~ \omega \, ,
\end{equation}
used to express $\hat{g}$, allow to construct past null infinity $\scri^-$ as the hypersurface $\{R=0\}$. We use the standard notation $i^0$ for spacelike infinity.

\section{The Couch-Torrence inversion} \label{CTSection}

The original construction by Couch and Torrence \cite{CoTo1983} in 1983, was a spatial inversion for a modified radial coordinate on the extreme Reissner-Nordström spacetime that turned out to be a global conformal isometry of the exterior region $\Ext$.
Introducing a radial coordinate centered on the horizon in $\Ext$:
\[ y = \frac{r-M}{M}\, \]
the metric $g$ can be expressed as:
\[ g = \frac{y^2}{(1+y)^2} \d t^2 - M^2 \frac{(1+y)^2}{y^2} \d y^2 - M^2 \left(1+y\right)^2 \d \omega^2  .\]
In these coordinates the \emph{Couch-Torrence inversion} is the map $\Phi: (t,y,\omega) \mapsto (t,\frac{1}{y},\omega)$. The pullback of the metric $g$ by $\Phi$ is given by:
\[\begin{aligned} \Phi^*g &= \frac{1}{y^2(1+\frac{1}{y})^2}\d t^2 - M^2y^2\left(\frac{1}{y}+1\right)^2\frac{\d y^2}{y^4}-M^2\left(\frac{1}{y}+1\right)^2\d \omega^2 =\frac{1}{y^2} g , \end{aligned}\]
which shows that $\Phi$ is a conformal isometry of the exterior region.
\begin{remark}
In their original article~\cite{CoTo1983}, Couch and Torrence work with the coordinate $x=-\frac{1}{y}$.
\end{remark}
In terms of the more usual coordinate $r$,  the transformation can be expressed as:
\begin{equation} \label{CT} \Phi(t,r,\omega) = (t, \frac{rM}{r-M}, \omega), \end{equation} and
\[ \Phi^* g = \frac{M^2}{(r-M)^2} g \, .\]
The Couch-Torrence inversion is an involution that exchanges the horizon $\{r=M\}$ and infinity $\{R=0\}$ in $\Ext$ and fixes every point on the photon sphere $\{r=2M\}$. 

%
%
%

Things are in fact much simpler when expressed in terms of the Regge-Wheeler variable $r_*$
[Eq.~\eqref{rstar}], which, we recall, was chosen centered on the photon sphere. Indeed:
\begin{eqnarray}
r_*\left(\frac{rM}{r-M}\right)  &=& \frac{M^2}{r-M} + 2M\log\left( \frac{M}{r-M}\right)-(r-M) \nonumber \\
&=&-\left(r-M +2M \log\left(\frac{r-M}{M}\right) - \frac{M^2}{r-M} \right) \nonumber \\
&=& - r_* (r) \, .
\end{eqnarray}
So the Couch-Torrence inversion can be simply stated as $r_*\mapsto -r_*$.

Something that does not seem to have been noticed in the literature is that the Couch-Torrence inversion is in fact an \emph{isometry} (and not just a conformal isometry) of the conformally compactified extreme Reissner-Nordström exterior region $(\widehat{\Ext},\hat{g})$ with conformal factor $\Omega=1/r$.
This can be seen directly, observing that
\begin{equation} \label{CTIsomRescMet}\Phi^*\hat{g}=\Phi^*(\Omega^2g)=(\Omega\circ \Phi)^2 \Phi^*g=\frac{(r-M)^2}{r^2M^2}\frac{M^2}{(r-M)^2}g=\hat{g},\end{equation}
%
%
%
%
%
\begin{remark}
One may wonder if $\Omega=\frac{1}{r}$ is the only conformal factor $\Omega$ for which the Couch-Torrence inversion $\Phi$ is an isometry. It is not, and those conformal factors that have this property have the general form:
\begin{equation}
\Omega = \frac{f(t,r_*,\omega )}{r},
\end{equation}
where $f$ is an arbitrary (positive) function that is even in $r_*$. This follows directly from the intermediate steps in Equation~\eqref{CTIsomRescMet}, from which we can see that $\Omega$ must satisfy:
\begin{equation} \label{Omega_for_isometry}
\frac{rM}{r-M}\Omega(t,\frac{rM}{r-M},\omega)=r\Omega(t,r,\omega),
\end{equation}
which translates to the fact that the scalar field $r\Omega$ is invariant under the Couch-Torrence inversion, i.e. is an even function in $r_*$.
Defining $x=r/M$ and $\lambda(x) = x \Omega(t,xM,\omega)$, Eq.~\eqref{Omega_for_isometry}
becomes
\begin{equation} \label{F_for_isometry}
 \lambda \left(\frac{x}{x-1}\right) = \lambda (x) .
\end{equation}
There are as many conformal factors $\Omega$ making $\Phi$ an isometry as there
are smooth positive solutions of this equation. The case $\Omega = 1/r$
corresponds to the trivial solution $\lambda (x) = \mathrm{const} = M^{-1}$ of
Eq.~\eqref{F_for_isometry}.
A whole family of solutions of
Eq.~\eqref{F_for_isometry} is
\[ \lambda (x) = \frac{a_1(x - 1) +  a_2 x^2}{b_1(x - 1) + b_2 x^2},
 \quad (a_1, a_2, b_1, b_2)\in \mathbb{R}^4, \quad
 (a_1,a_2)\neq (0,0),  \quad
 (b_1,b_2)\neq (0,0) .
\]
For instance, for $(a_1, a_2, b_1, b_2) = (0, 1, 1, 0)$, we get
\[ \Omega = \frac{M r}{r - M} , \]
while $(a_1, a_2, b_1, b_2) = (1, 0, 0, 1)$ yields
\[ \Omega = \frac{M^2(r - M)}{r^3} . \]
\end{remark}
%
We also address the question of how one could interpret the Couch-Torrence inversion
on $\Int$ --- the interior of the black hole. It is apparent that it is not an endomorphism of $\Int$,
given that $rM/(r-M) < 0$ for $0<r<M$,
but we can instead view the coordinate expression~\eqref{CT} as defining a map from  $\Int$ into a manifold $\mathcal{N}=\R_v\times(-\infty,0)_r\times S^2$ equipped with the metric
 \[ g_{\mathcal{N}} = F(r)\d v^2 +2\d v \d r -r^2 \d \omega^2. \]
Using the more appropriate outgoing Eddington-Finkelstein coordinates on $\textrm{Int}$,
this map can be expressed as
 $\Phi : (\Int,g) \rightarrow (\mathcal{N},g_{\mathcal{N}})$,
 $(u,r,\omega)\mapsto (u+2r^*(r),\frac{rM}{r-M},\omega)$. Furthermore, one has:
 \[\Phi^*g_{\mathcal{N}}=\frac{M^2}{(r-M)^2}g. \] $\Phi$ is therefore a conformal isometry between $\Int$ and $\mathcal{N}$. Performing the change of coordinate $r'=-r$ in $\mathcal{N}$, we can identify $(\mathcal{N},g_{\mathcal{N}})$ with an extreme Reissner-Nordström black-hole with negative mass $-M$ and charge $Q=\pm M$ (expressed in ingoing Eddington-Finkelstein coordinates~\eqref{ef-ingoing})\footnote{This property was in fact already observed in \cite{LVK2014}.}.  In terms of $r'$, $\Phi^{-1}: \mathcal{N} \rightarrow \textrm{Int}$, $(t,r',\omega) \mapsto (t, \frac{r'M}{r'+M},\omega)$. As before, the conformal isometry becomes an isometry if both $g$ and $g_{\mathcal{N}}$ are conformally rescaled by $\Omega=\frac{1}{r}$.

These properties of the Couch-Torrence inversion, are summarised in the following theorem.
\begin{theorem} \label{ThmIsomCTModCT}
The Couch-Torrence inversion \eqref{CT} is an \textbf{isometry} of the compactified exterior of the extreme Reissner-Nordström spacetime $\widehat{\Ext}$ with conformal factor $R=1/r$, i.e.
\[ \Phi^* \hat{g} = \hat{g} \, ,~\mbox{with } \hat{g} = \frac{1}{r^2} g \, .\]
It fixes the photon sphere and exchanges the future event horizon $\scrh^+$ and the future null infinity $\scri^+$, as well as the past event horizon $\scrh^-$ and the past null infinity $\scri^-$.

Interpreting the coordinate expression of the Couch-Torrence inversion as a diffeomorphism from the interior of the black hole into another spacetime leads to an isometry from $(\mathrm{Int},\hat{g})$ onto a full negative mass extreme Reissner-Nordström spacetime with mass $-M$ and charge $Q=\pm M$, whose metric has been conformally rescaled by $1/r^2$. The curvature singularities of both spacetimes are each other's images under $\Phi$ and the extreme Reissner-Nordström horizon corresponds to the infinity of the negative mass extreme Reissner-Nordström spacetime.
\end{theorem}
\begin{corollary} \label{ConfWEqInvarianceCT}
As a consequence, both the d'Alembertian $\square_{\hat{g}}$ and the scalar curvature $\mathrm{Scal}_{\hat{g}}$ outside the extreme Reissner-Nordström black hole, are invariant under the Couch-Torrence inversion (this can be easily checked by direct calculations). Therefore the conformal d'Alembertian
\[ \square_{\hat{g}}  + \frac{1}{6} \mathrm{Scal}_{\hat{g}}  \]
is invariant under the Couch-Torrence inversion.
\end{corollary}
It is natural to ask if there could be an isometry of the whole domain of outer communication of a black hole with a non-degenerate horizon, that would exchange the horizon and conformal infinity. There are two arguments against this.
First, in the stationary case, it is a general feature of
stationary asymptotically flat spacetimes that the future null infinity $\scri^+$
of the $1/r$-compactification
is a \emph{degenerate} Killing horizon with respect to the vector field $\partial/\partial u$
of the outgoing Eddington-Finkelstein coordinates, which is a Killing
vector of the conformal metric $\hat{g} = \Omega^2 g$ as soon as $\Omega$ is a function of $r$ only. On the other side, the future event horizon $\scrh^+$ is a Killing horizon
with respect to the vector field $\partial/\partial v$ of the ingoing Eddington-Finkelstein
coordinates, which is a Killing vector of both $g$ and $\hat{g}$
for $\Omega = \Omega(r)$. It is easy to see that the surface gravity $\kappa$ of $\scrh^+$
is conformally invariant. If $\Phi$ were to be an isometry mapping $\scrh^+$
to $\scri^+$, it could not map a non-degenerate Killing horizon $(\kappa\neq 0)$
to a degenerate one $(\kappa= 0)$; hence  $\scrh^+$ has to be degenerate.
Second and more generally, spacelike infinity is a conformal singularity unlike the bifurcation sphere. More precisely, null geodesics along the future or past horizons will reach the bifurcation sphere with finite affine parameters, in contrast, null geodesics on $\scri^\pm$ are complete and these would be exchanged by the transformation.

\section{Peeling at the extreme Reissner-Nordström horizon using the Couch-Torrence inversion} \label{PeelingRNE}
Since the Couch-Torrence inversion on $\widehat{\Ext}$ is an isometry that exchanges the horizon and conformal infinity $\{R=0\}$, knowledge about either of them will translate to information about the other. Our first goal is to establish a peeling property at the future null infinity $\scri^+$. This feature has been studied at infinity in other spacetimes and, following~\cite{MaNi2009}, we will show that the result subsists in Reissner-Nordström spacetime (by simply observing that the estimates can be performed as in the Schwarzschild case).  We will then translate this into a peeling property at the degenerate horizon $\scrh^+$.
 We explain the essential steps of the proof for the convenience of the reader.

\subsection{Peeling at infinity on the extreme Reissner-Nordström metric}

In this section we extend the work of \cite{MaNi2009} to the extreme Reissner-Nordström spacetime. The goal is to characterise the regularity at $\scri^+$ of the solution to the conformal wave equation in terms of the regularity and decay of the initial data. Once the regularity is known in an arbitrarily small neighbourhood of $i^0$, standard results allow to propagate it in a complete neighbourhood of $\scri^+$ (in fact to the full domain $\widehat{\Ext}$) provided the initial data have the same degree of smoothness (see for example Friedrich \cite{HFri2004}). In order to control the regularity within a small neighbourhood of $i^0$, we prove energy estimates both ways and at all orders between the part of $\scri^+$ and the part of the $\{t=0\}$ Cauchy hypersurface that are contained in this neighbourhood. The energy current is associated with a Morawetz vector field adapted to the geometry and the various levels of regularity are obtained by considering the energy of successive partial derivatives of the field. This energy current satisfies an approximate conservation law as we approach $i^0$. The size of the neighbourhood of $i^0$ is adapted so as to allow a control of the error terms by the energy on the slices of a well-chosen foliation; the estimates then follow by Grönwall's inequality. All the estimates are established for solutions associated with smooth and compactly supported data. Given the linear nature of the equation, their validity then naturally extends by density to the function spaces constructed by completing the space of smooth compactly supported functions in the norms defined by the energies.

We work in the exterior block $\Ext$ equipped with the \emph{unphysical} metric $\hat{g}=R^2g$ given by~\eqref{CompExtRNMetEF} expressed in outgoing Eddington-Finkelstein coordinates with an inversion in $r$
\[ R = \frac{1}{r} \, ,~ u = t-r_* \, ,~ \omega \, .\]
The inverse metric is
\[ \hat{g}^{-1} = -2\p_u\p_R-R^2(1-MR)^2\p_R^2-\eth\bar\eth \, .\]
The scalar curvature of $\hat{g}$ has the form
\begin{equation}
\mathrm{Scal}_{\hat{g}} = 12 M R (MR - 1 )
\end{equation}
and the induced $4$-volume form reads
\begin{equation} \label{4Vol}
\d^4\mathrm{Vol} = \d u \wedge \d R \wedge \d^2 \omega \, .
\end{equation}
We remark that this specifies our global choice of orientation to be that of the basis $(\partial_u,\partial_R,\partial_\theta,\partial_\phi)$.

Let us denote by $\hat{\nabla}$ the Levi-Civita connection induced by $\hat{g}$. Since the scalar curvature of the physical metric vanishes, $\mathrm{Scal}_{g}=0$, it follows from~\eqref{CWaveScalar} that a scalar field $\psi$ satisfies the \emph{wave equation} for $g$ outside the black hole
\begin{equation} \label{WE}
\square_g \psi = 0 \, ,
\end{equation}
if and only if $\phi := R \psi$ satisfies
\begin{equation} \label{CWERNE}
\square_{\hat{g}} \phi + 2MR(MR-1) \phi = 0 \, ,
\end{equation}
where the d'Alembertian for $\hat{g}$ is given by
\begin{equation} \label{dAlembertianRNE}
\square_{\hat{g}} f = - 2 \frac{\partial^2\,f}{\partial u\partial R}
- \frac{\partial}{\partial R} \left( R^2(1-MR)^2 \frac{\partial f}{\partial R } \right)
 -\Delta_{S^2} f \, .
\end{equation}

Study of the peeling commences with the choice of an appropriate energy current. As in Minkowski and Schwarzschild spacetimes, our choice will be associated with the family of observers given by the Morawetz vector field
\begin{equation} \label{Morawetz}
K = u^2 \partial_u - 2 (1+uR) \partial_R \, .
\end{equation}
obtained, as in the Schwarzschild case, by transposition of the formula for the Morawetz field in Minkowski space time, expressed in outgoing light-cone coordinates. It satisfies
\begin{equation} \label{NormMorawetz}
\hat{g} (K,K) = u^2 \left( 4 + 4Ru + R^2u^2 (1-MR)^2 \right) \, .
\end{equation}
Since $4 + 4Ru + R^2u^2 = (2+Ru)^2$, \eqref{NormMorawetz} is positive in a neighbourhood of $i^0$ (see Lemma \ref{ApproxCloseI0} below). This and the expression \eqref{Morawetz} entail that $K$ is timelike and future-oriented in a neighbourhood of $i^0$.
It is \emph{not} a Killing (or even conformal Killing) vector field of $\hat{g}$, as can be seen from its Killing form:
\begin{equation} \label{MorawetzKillingFormRNE}
\hat{\nabla}_{(a} K_{b)} \d x^a \d x^b = 2 M R^2 \left(
Ru(1 - M R) - 2 M R + 3 \right) \mathrm{d} u\otimes \mathrm{d} u \, .
\end{equation}
%
Consider now the stress-energy tensor for the free wave equation $\Box_{\hat{g}}\phi=0$ on the compactified spacetime
\[ T_{ab} (\phi) = \hat{\nabla}_a {\phi} \hat{\nabla}_b {\phi} - \frac12 \langle \hat{\nabla} {\phi} , \hat{\nabla} {\phi} \rangle_{\hat{g}} \hat{g}_{ab} \]
and define 
\begin{equation} \label{EnCurrentRNE}
J_a (\phi) := K^b T_{ab} (\phi) \, .
\end{equation}

Our analysis will concentrate on a neighbourhood of $i^0$ defined for $u_0 \ll -1$,
\begin{equation}
\Omega_{u_0} = \{ t\geq 0 \} \cap \{ u<u_0 \} \, ,
\end{equation}
that we foliate with the hypersurfaces
\begin{equation}
\mathcal{H}_{s,u_0} := \{ u=-sr_* \, , ~ u < u_0 \}, \quad 0\leq s \leq 1 \, ,
\end{equation}
where $\mathcal{H}_{0,u_0}$ is considered as the limit of the $\mathcal{H}_{s,u_0}$ hypersurfaces as $s\rightarrow 0$ and is in fact
\[ \mathcal{H}_{0,u_0} = \scri^+ \cap \{ u < u_0 \} =: \scri^+_{u_0} \, .\]
Having a regular slicing between $\scri^+_{u_0}$ and $\mathcal{H}_{1,u_0}=\{t=0\}\cap\{u < u_0\}$ gives a convenient way of controlling the energies on either hypersurface in terms of the other via Grönwall estimates. Another important hypersurface is part of the future boundary of $\Omega_{u_0}$
\begin{equation}
\mathcal{S}_{u_0} = \{ t \geq 0\} \cap \{ u = u_0\} \, .
\end{equation}

We orient each of the $\mathcal{H}_{s,u_0}$ using the future pointing normal, i.e. in the direction of decreasing $s$.
The energy flux through any slice $\mathcal{H}_{s,u_0}$ is given by:
\begin{eqnarray}
\mathcal{E}_{\mathcal{H}_{s,u_0}}(\phi)&=& \int_{\mathcal{H}_{s,u_0}} \star J(\phi), \nonumber \\
 &=& \int_{ ]-\infty , u_0 [ \times S^2} \bigg( R^2 (1-MR)^2 u^{2}  \frac{\partial\,\phi}{\partial u} \frac{\partial\,\phi}{\partial R} + u^{2} \left( \frac{\partial\,\phi}{\partial u} \right)^{2} \nonumber \\
&& + \frac{R^2 (1-MR)^2}{2s} \left( (Ru)^2\left(1-MR\right)^2 + 2(2-s) Ru + 2 (2-s)\right) \left( \frac{\partial\,\phi}{\partial R} \right)^{2} \nonumber \\
&&+ \left( R u + \frac{(Ru)^2 (1-MR)^2}{2s} + 1 \right) \left\vert \nabla_{S^2} \phi \right\vert^{2}  \bigg) \d u \wedge \d^2 \omega \\
&=&  \int_{ ]-\infty , u_0 [ \times S^2} \bigg( R^2 (1-MR)^2 u^{2}  \frac{\partial\,\phi}{\partial u} \frac{\partial\,\phi}{\partial R} + u^{2} \left( \frac{\partial\,\phi}{\partial u} \right)^{2} \nonumber \\
&& + \frac{Rr_* (1-MR)^2}{2}\frac{R}{\vert u \vert} \left( (R\vert u\vert)^2\left(1-MR\right)^2 - 2(2-s) R\vert u\vert + 2 (2-s)\right) \left( \frac{\partial\,\phi}{\partial R} \right)^{2} \nonumber \\
&&+ \left( R\vert u \vert \left( -1 + \frac{R r_* (1-MR)^2}{2} \right) + 1 \right) \left\vert \nabla_{S^2} \phi \right\vert^{2}  \bigg) \d u \wedge \d^2 \omega \, . \label{EnergyHsSecondForm}
\end{eqnarray}
where $\star$ is the Hodge dual defined on $1$-forms by the identity:
\begin{equation}
\alpha \wedge \star \beta = g(\alpha,\beta) \d^4\textrm{Vol} =g^{ab}\alpha_a\beta_b \d^4\textrm{Vol}.
\end{equation}
%
%
\begin{remark}
Alternatively, one can use the formula:
\[\mathcal{E}_{\mathcal{H}_{s,u_0}}(\phi)=\int_{ ]-\infty , u_0 [ \times S^2} J_a (\phi) \tilde{n}^a \, \tilde{l} \lrcorner \d^4\mathrm{Vol},\]
where $\tilde{n}$ is a future-oriented normal vector field to $\mathcal{H}_{s,u_0}$ and $\tilde{l}$ a transverse vector field to all $\mathcal{H}_{s,u_0}$ such that $\hat{g}(\tilde{l},\tilde{n})=1$.
A future-oriented normal vector field to $\mathcal{H}_{s,u_0}$ is easily obtained from:
\[\begin{aligned} n := \hat{g}^{-1} ( \d (u/ r_*) ) &=\hat{g}^{-1} \left( \frac{1}{r_*} \d u -\frac{u}{r_*^2} \frac{1}{F} \d r \right)  = - \frac{1}{r_*}\left( 1-s \right) \partial_R - \frac{u}{(1-MR)^2(r_*R)^2} \partial_u \, . \end{aligned} \]
Note that on $\mathcal{H}_{0,u_0}$, $n$ reduces to $-u\partial_u$  which is future oriented for $u<0$, i.e. where the foliation makes sense. It is straightforward to check that:
\[ l = \frac{(1-MR)^2(r_*R)^2}{u} \partial_R \, , \]
satisfies $\hat{g}(l,n)=1$.
In order to simplify the expressions of the two vectors fields, we multiply $n$ by $- \frac{(1-MR)^2(r_*R)^2}{u}$ and $l$ by $- \frac{u}{(1-MR)^2(r_*R)^2}$ leading to:
\begin{equation} \label{lnHs}
\tilde{n} = \partial_u + (1-MR)^2 R^2\frac{r_*}{u}\left( 1-s \right) \partial_R \, ,~ \tilde{l} = - \partial_R \, .
\end{equation}
\end{remark}
We also orient $\mathcal{S}_{u_0} $ using the future oriented normal, i.e. $-\partial_R$ and we have
\begin{eqnarray}
\mathcal{E}_{\mathcal{S}_{u_0}}(\phi)&=& \int_{\mathcal{S}_{u_0}} \star J(\phi), \nonumber \\
&=& -\frac12 \int_{\mathcal{S}_{u_0}} \bigg( \left( 4 +4 R u + (Ru)^2 (1-MR)^2 \right) \left(\frac{\partial\,\phi}{\partial R}\right)^{2} \nonumber \\
&& \hspace{0.8in} + u^2 \vert \nabla_{S^2} \phi \vert^2 \bigg) \d R \wedge \d^2 \omega \, ,
\end{eqnarray}
which is non negative by the dominant energy condition since $\mathcal{S}_{u_0}$ is a null hypersurface.

In order to derive our fundamental estimates, we will assume that $\phi$ is a solution to~\eqref{CWERNE}, with \emph{smooth compactly supported initial data}. We will then be able to extend the estimates by density to the completion in the energy norm. We give a simplified equivalent expression of the energy flux across $\mathcal{H}_{s,u_0}$. This is identical to the corresponding result in the Schwarzschild geometry~\cite[Lemma 4.1]{MaNi2009}. We give the details of the proof since they will be useful when extending the results to more general geometries in Section \ref{LooseEndCT}.
\begin{proposition}\label{EnEqFormRNE}
Using the fact that $Ru$ remains bounded on $\mathcal{H}_{s,u_0}$ one has the following equivalence uniformly in $s \in [0,1]$, provided $u<u_0<0$ and $\vert u_0\vert $ sufficiently large :
\begin{equation} \mathcal{E}_{\mathcal{H}_{s,u_0}}(\phi)\simeq  \int_{ ]-\infty , u_0 [ \times S^2} \bigg(  u^{2}  \left( \frac{\partial\,\phi}{\partial u} \right)^{2} + \frac{R}{\vert u \vert} \left( \frac{\partial\,\phi}{\partial R} \right)^{2} + \left\vert \nabla_{S^2} \phi \right\vert^{2}  \bigg) \d u \wedge \d^2 \omega \, .\label{EnergyHsRNE}
\end{equation}
\end{proposition}
The proof uses the following obvious results.
\begin{lemma}\label{ApproxCloseI0}
Let $\varepsilon >0$, then one can find $u_0<0$, $\vert u_0\vert $ large enough, such that in $\Omega_{u_0}$,
\[ 1 < Rr_* < 1+\varepsilon  \, ,~ 0 < R\vert u\vert  < 1+\varepsilon  \, ,~ 1-\varepsilon < (1-MR)^2 <1 \, .\]
\end{lemma}
\noindent{\bf Proof of Proposition \ref{EnEqFormRNE}.} In $\Omega_{u_0}$, we have
\[ 1 + R \vert u \vert \left( \frac{Rr_* (1-MR)^2 }{2}-1\right) \geq 1 - \frac{(1+\varepsilon)R \vert u\vert}{2} \geq \frac{1-\varepsilon'}{2} \]
where $\varepsilon' \rightarrow 0$ as $\varepsilon \rightarrow 0$. The coefficient of $(\partial_R \phi )^2$ involves the expression
\[ P( R\vert u \vert ) := (R\vert u\vert)^2\left(1-MR\right)^2 - 2(2-s) R\vert u\vert + 2 (2-s) \]
multiplied by a factor that is equivalent to $\frac12 \frac{R}{\vert u \vert}$. Restricting to $\mathcal{H}_{s,u_0}$ for $0<s\leq 1$ fixed, this becomes
\begin{eqnarray*}
P(R\vert u \vert ) &=& (Rr_*)^2s^2\left(1-MR\right)^2 - 2(2-s) Rr_*s + 2 (2-s) \\
&\geq & (1-\varepsilon) s^2 + 2(2-s) (1- (1+\varepsilon ) s)\\
&&= s^2 + 2 (2-s)(1-s) -\varepsilon (s^2 + 2s(2-s) ) \\
&& = 3 (s-1)^2 +1 - \varepsilon (s^2 + 2s(2-s) ) \\
&& \geq 1 - \varepsilon''
\end{eqnarray*}
where $\varepsilon'' \rightarrow 0$ as $\varepsilon \rightarrow 0$. And for $s =0$, $P(R\vert u \vert ) = 4$.
The coefficient of $\frac{\partial\,\phi}{\partial u} \frac{\partial\,\phi}{\partial R}$ is $R^2 (1-MR)^2 u^{2}$ and
\[ R^2 (1-MR)^2 u^{2} = \sqrt{\frac{R}{\vert u \vert}} \vert u \vert (1-MR)^2 (R\vert u\vert )^{3/2} \\
\leq  (1+\varepsilon)^{3/2} \sqrt{\frac{R}{\vert u \vert}} \vert u \vert \, . \]
Hence,
\begin{eqnarray*}
\left\vert R^2 (1-MR)^2 u^{2} \frac{\partial\,\phi}{\partial u} \frac{\partial\,\phi}{\partial R} \right\vert \leq (1+\varepsilon)^{3/2} \left( \frac{\lambda^2}{2} \frac{R}{\vert u \vert} \left( \frac{\partial\,\phi}{\partial R}\right)^2 + \frac{1}{2\lambda^2} u^2 \left( \frac{\partial\,\phi}{\partial u}\right)^2 \right) \, .
\end{eqnarray*}
If we can choose $\lambda >0$ such that $\frac{\lambda^2}{2} < \frac12$ and $\frac{1}{2\lambda^2} <1$, this proves the proposition. This is obviously possible since the two inequalities reduce to $\frac{1}{\sqrt{2}} < \lambda < 1 $.\qed

Note that for $s=0$, the equivalence is in fact an equality as can be seen from \eqref{EnergyHsSecondForm}:
\begin{equation} \label{EnergyScriRNE}
{\cal E}_{\scri^+_{u_0}} (\phi ) = \int_{\scri^+_{u_0}} \bigg( u^2 \bigg( \frac{\partial\,\phi}{\partial u} \bigg)^2 + \vert \nabla_{S^2} \phi\vert ^2 \bigg) \d u \wedge \d^2 \omega \, .
\end{equation}

The main tool for obtaining our basic energy estimates is Stokes theorem applied to
\[ \d(\star J(\phi))=\hat{\nabla}^aJ_a(\phi)d^4\mathrm{Vol}. \]
Integrating over $\Omega_{u_0}$, with $\phi \in \mathcal{C}^\infty(\widehat{\Ext})$ supported away from $i^0$, leads to the
fundamental energy identity
\begin{equation} \label{FEI}
\mathcal{E}_{\scri^+_{u_0}} (\phi) + \mathcal{E}_{\mathcal{S}_{u_0}} (\phi ) - \mathcal{E}_{\mathcal{H}_{1,u_0}} (\phi ) = \int_{\Omega_{u_0}} \hat{\nabla}^a J_a (\phi) \d^4\mathrm{Vol} \, .
\end{equation}
Since $K^a$ is not a Killing vector field, this is only an approximate conservation law and the error terms are determined by
\begin{eqnarray}
\hat{\nabla}^a J_a (\phi) &= & \hat{\nabla}^{(a} K^{b)} T_{ab} - \frac{1}{6} \mathrm{Scal}_{\hat{g}} \phi \nabla_{K} \phi  \nonumber \\
&= & -2 \, {\left(R^{2} M^{2} - R M\right)} u^{2} \phi \frac{\partial\,\phi}{\partial u} + 4 \, {\left(R^{2} M^{2} - R M + {\left(R^{3} M^{2} - R^{2} M\right)} u\right)} \phi \frac{\partial\,\phi}{\partial R}  \nonumber \\
&& - 2 \, {\left(2 \, R^{3} M^{2} - 3 \, R^{2} M + {\left(R^{4} M^{2} - R^{3} M\right)} u\right)} \left( \frac{\partial\,\phi}{\partial R}\right)^{2}. \label{ApproxConsLawRNE}
\end{eqnarray}
Provided $u_0\ll-1$ is large enough in absolute value, we have
\begin{equation} \label{LocEstErrTermsRNE}
\vert \hat{\nabla}^a J_a (\phi )\vert \lesssim \phi^2 + u^2 (\partial_u \phi )^2 + R^2 (\partial_R \phi )^2 \, .
\end{equation}
The integral on $\Omega_{u_0}$ on the right hand-side of \eqref{FEI} can be decomposed into an integral in $s$  of integrals on the slices $\mathcal{H}_{s,u_0}$. This is done by choosing an appropriate foliation chart that we describe by an identifying vector field $\nu$, transverse to all $\mathcal{H}_{s,u_0}$, and such that $\nu (s) = -1$ (so that $\nu$ be future oriented), for instance
\begin{equation} \label{VectVHs}
\nu=\frac{(r_*R)^2(1-MR)^2}{u}\frac{\partial}{\partial R} = l \, .
\end{equation}
The $4$-volume measure can then be decomposed as follows
\begin{equation} \label{Splitting4Vol}
\d^4\mathrm{Vol} = -\d s \wedge \left( l \lrcorner \d^4\mathrm{Vol} \right)
\end{equation}
where $\left( l \lrcorner \d^4\mathrm{Vol} \right)\vert_{\mathcal{H}_{s,u_0}}$ is a $3$-measure on $\mathcal{H}_{s,u_0}$ that is equivalent to
\begin{equation} \label{Equiv3Vol}
\frac{1}{ u } \partial_R \lrcorner \d^4\mathrm{Vol} = \frac{1}{ \vert u\vert} \d u \wedge \d^2 \omega \, .
\end{equation}
Hence
\begin{align}
\int_{\Omega_{u_0}} \hat{\nabla}^a J_a (\phi) \d^4\mathrm{Vol} &= \int_{[0,1]} \left( \int_{\mathcal{H}_{s,u_0}}\hat{\nabla}^a J_a (\phi) \, l \lrcorner \d^4\mathrm{Vol} \right) \d s \nonumber \\
&\simeq \int_{[0,1]} \left( \int_{]-\infty ,u_0 [ \times S^2}  \hat{\nabla}^a J_a (\phi) \frac{1}{\vert u\vert} \d u \d^2 \omega \right) \d s \, .
\end{align}
The control of the error terms \eqref{LocEstErrTermsRNE} then gives
\begin{align*}
\int_{]-\infty ,u_0 [ \times S^2}  \hat{\nabla}^a J_a (\phi) \frac{1}{\vert u\vert} \d u \d^2 \omega &\lesssim \int_{]-\infty ,u_0 [ \times S^2}  \left( \phi^2 + u^2 (\partial_u \phi )^2 + R^2 (\partial_R \phi )^2 \right) \frac{1}{\vert u\vert} \d u \d^2 \omega \\
&\lesssim \int_{]-\infty ,u_0 [ \times S^2}  \left( \phi^2 + u^2 (\partial_u \phi )^2 + R^2 (\partial_R \phi )^2 \right) \d u \d^2 \omega  \, .
\end{align*}
The Poincaré-type estimate obtained in \cite{MaNi2009} (see Lemma 4.2 and Corollary 4.1 of that paper) implies
\[  \int_{]-\infty ,u_0 [ \times S^2} \phi^2 \d u \d^2 \omega \lesssim \int_{]-\infty ,u_0 [ \times S^2} u^2 (\partial_u\phi)^2 \d u \d^2 \omega \]
and it follows that
\[ \int_{]-\infty ,u_0 [ \times S^2}  \hat{\nabla}^a J_a (\phi) \frac{1}{\vert u\vert} \d u \d^2 \omega \lesssim \mathcal{E}_{\mathcal{H}_{s,u_0}} (\phi) \, .\]
This allows us to obtain estimates both ways between $\mathcal{E}_{\scri^+_{u_0}} (\phi) + \mathcal{E}_{\mathcal{S}_{u_0}} (\phi )$ and $\mathcal{E}_{\mathcal{H}_{1,u_0}} (\phi )$ using Grönwall's Lemma.

In order to study higher order estimates we commute derivatives into the equation, however, this requires a little more care because new terms will appear, some of which do not vanish at infinity. This can be seen typically for the equation for $\partial_R \phi$ that is obtained by commuting $\partial_R$ into \eqref{CWERNE}. The resulting equation reads
\[ \square_{\hat{g}} (\partial_R \phi ) = \left[ \square_{\hat{g}} , \partial_R \right] \phi - 2MR(MR-1) \partial_R \phi - 4 M^2 R \phi + 2M \phi \, .\]
Using the expression \eqref{dAlembertianRNE} of $\square_{\hat{g}}$, we see that
\begin{align*}
\left[ \square_{\hat{g}} , \partial_R \right] &= \partial_R (R^2(1-MR)^2) \frac{\partial^2}{\partial R ^ 2} + \partial_R (2R(1-MR)^2) \frac{\partial}{\partial R} \\
&= \left( 2R (1-MR)^2 - 2MR^2 (1-MR) \right) \frac{\partial^2}{\partial R ^ 2} \\
& + \left( 2 (1-MR)^2 - 2MR (1-MR) \right) \frac{\partial}{\partial R} \, .
\end{align*}
Hence, if $\phi$ is a solution to \eqref{CWERNE}, $\partial_R \phi$ satisfies
\begin{equation}
\begin{aligned}
\square_{\hat{g}} (\partial_R \phi ) &=\left( 2R (1-MR)^2 - 2MR^2 (1-MR) \right) \frac{\partial}{\partial R} (\partial_R \phi ) \\
& + \left( 2 (1-MR)^2 - 2MR (1-MR) \right) \partial_R \phi  \\
&- 2MR(MR-1) \partial_R \phi - 4 M^2 R \phi + 2M \phi \, . \label{RNEEqDRphi}
\end{aligned}
\end{equation}
We can obtain an approximate conservation law for $\partial_R \phi$ by defining the energy current
\[ J_a (\partial_R \phi ) := K^b T_{ab} (\partial_R \phi )\]
and calculating its divergence using equation \eqref{RNEEqDRphi}
\begin{equation}
\begin{aligned}
\hat{\nabla}^a J_a (\partial_R \phi ) = & \, \hat{\nabla}^{(a} K^{b)} T_{ab} (\partial_R \phi ) \\
&+ \left( \hat{\nabla}_K (\partial_R \phi ) \right) \left( 2R (1-MR)^2 - 2MR^2 (1-MR) \right) \frac{\partial}{\partial R} (\partial_R \phi ) \\
& + \left( \hat{\nabla}_K (\partial_R \phi ) \right)\left( 2 (1-MR)^2 - 2MR (1-MR) \right) \partial_R \phi  \\
& + \left( \hat{\nabla}_K (\partial_R \phi ) \right) \left( - 2MR(MR-1) \partial_R \phi - 4 M^2 R \phi + 2M \phi \right) \, . \label{ApproxConsLawRNEDRphi}
\end{aligned}
\end{equation}
Both the terms $2 (1-MR)^2 \partial_R \phi $ and $2M \phi$ in the right hand-side of \eqref{RNEEqDRphi} have no decay at infinity and the corresponding error terms in the approximate conservation law \eqref{ApproxConsLawRNEDRphi} for $\partial_R \phi$, once integrated on $\mathcal{H}_{s,u_0}$, cannot be directly controlled by the energy for either $\partial_R \phi$ or $\phi$. However, the geometry of our foliation is particular since the leaves are getting closer and closer as we approach $i^0$. When we write the $4$-volume measure as the exterior product of $\d s$ and a $3$-measure on $\mathcal{H}_{s,u_0}$ (see \eqref{Splitting4Vol} and \eqref{Equiv3Vol}) a factor $1/\vert u \vert$ appears in front of the integral on the spacelike slices. Hence, what needs to be estimated by the energy is the integral over $\mathcal{H}_{s,u_0}$ of the error terms divided by $\vert u \vert$. For the fundamental estimates above, we did not use this $1/\vert u \vert$ factor but now it becomes crucial. We show here how the estimate is done for $2M \phi$, whose corresponding error term is $2M\phi \hat{\nabla}_K \phi $~:
\begin{align*} \int_{]-\infty , u_0 [ \times S^2} 2M\phi \hat{\nabla}_K \phi \frac{1}{\vert u \vert} \d u \d^2 \omega &= \int_{]-\infty , u_0 [ \times S^2} 2M\phi (u^2 \partial_u \phi -2 (1+uR) \partial_R \phi ) \frac{1}{\vert u \vert} \d u \d^2 \omega \\
&\lesssim \int_{]-\infty , u_0 [ \times S^2} \left( \phi^2 + u^2 (\partial_u \phi)^2 + \frac{1}{\vert u\vert} \vert \phi \vert \vert\partial_R \phi\vert \right) \d u \d^2 \omega \, .
\end{align*}
We now observe that
\[ \frac{1}{\vert u \vert } = \frac{1}{\sqrt{s} r_*} \frac{1}{\sqrt{\vert u \vert}} \simeq \frac{1}{\sqrt{s}} \sqrt{\frac{R}{\vert u \vert}} \, .\]
Hence,
\[ \int_{]-\infty , u_0 [ \times S^2} 2M\phi \hat{\nabla}_K \phi \frac{1}{\vert u \vert} \d u \d^2 \omega \lesssim \frac{1}{\sqrt{s}} \int_{]-\infty , u_0 [ \times S^2} \left( \phi^2 + u^2 (\partial_u \phi)^2 + \frac{R}{\vert u\vert} \left( \partial_R \phi \right)^2 \right)\d u \d^2 \omega \]
and the Poincaré estimate from \cite{MaNi2009} then gives
\begin{align*}
\int_{]-\infty , u_0 [ \times S^2} \hat{\nabla}_K \phi 2M\phi \frac{1}{\vert u \vert} \d u \d^2 \omega &\lesssim \frac{1}{\sqrt{s}} \int_{]-\infty , u_0 [ \times S^2} \left( u^2 (\partial_u \phi)^2 + \frac{R}{\vert u\vert} \left( \partial_R \phi \right)^2 \right)\d u \d^2 \omega \\
&\lesssim \frac{1}{\sqrt{s}} \mathcal{E}_{\mathcal{H}_{s,u_0}} (\phi ) \, .
\end{align*}
The term $2 (1-MR)^2 \partial_R \phi $ is treated similarly to obtain
\[ \int_{]-\infty , u_0 [ \times S^2} \hat{\nabla}_K \phi 2 (1-MR)^2 \partial_R \phi \frac{1}{\vert u \vert} \d u \d^2 \omega \lesssim \frac{1}{\sqrt{s}} \mathcal{E}_{\mathcal{H}_{s,u_0}} (\partial_R \phi ) \, . \]
The other error terms exhibit decay and can be controlled as was done for the fundamental estimate. Finally we obtain
\[ \int_{]-\infty ,u_0 [ \times S^2}  \nabla^a J_a (\partial_R \phi) \frac{1}{\vert u\vert} \d u \d^2 \omega \lesssim \frac{1}{\sqrt{s}} \left( \mathcal{E}_{\mathcal{H}_{s,u_0}} (\phi) + \mathcal{E}_{\mathcal{H}_{s,u_0}} (\partial_R \phi) \right) \]
and since the function $1/\sqrt{s}$ is integrable on $(0,1]$, we obtain estimates both ways for the sum of the energies of $\phi$ and $\partial_R \phi$ between $\mathcal{S}_{u_0} \cup \scri^+_{u_0}$ and $\mathcal{H}_{1,u_0}$. For more details, see \cite{MaNi2009} since everything goes through as it does in the Schwarzschild case.

As a consequence, all the theorems on the peeling for massless scalar fields at infinity on Schwarz\-schild's spacetime obtained in \cite{MaNi2009}, extend without modification to the extreme Reissner-Nordström geometry. We formulate one theorem that is in a sense the most complete result since it involves all directional derivatives, but the results available from \cite{MaNi2009} are more precise and detailed.
\begin{theorem} \label{ThmPeelingRNEInfinity}
Let $(u,R,\omega)$ be the outgoing Eddington-Finkelstein coordinates
of the extreme Reissner-Nordström exterior $\mathrm{Ext}$,
$u_0 \ll -1$, $k\in \mathbb{N}$ and $\phi$ a solution to the conformal wave equation~\eqref{CWERNE}. Then $\mathcal{E}_{\scri^+_{u_0}}(\partial^q_R\nabla^p_{S^2}\phi) + \mathcal{E}_{\mathcal{S}_{u_0}}(\partial^q_R\nabla^p_{S^2}\phi) < +\infty$ for all $p,q\in \mathbb{N}$, $p+q\leq k$ if and only if the initial data $(\phi_0, \phi_1)$ on the spacelike slice $\{t=0\}$ is chosen in the completion of $C^\infty_0([-u_0,+\infty[_{r_*}\times S^2)\times C^\infty_0([-u_0,+\infty[_{r_*}\times S^2)$ in the norm:
\begin{equation} \left\lVert\begin{pmatrix}\phi_0\\\phi_1\end{pmatrix}\right\rVert_k^2=\sum_{p+q\leq k} \mathcal{E}_{\mathcal{H}_{1,u_0}}\left(L^q\nabla^p_{S^2}\begin{pmatrix}\phi_0\\\phi_1\end{pmatrix} \right), \end{equation}
where we replace $\phi$ by $\phi_0$ and $\partial_t\phi$ by $\phi_1$ in the expression of the energy flux $\mathcal{E}_{\mathcal{H}_{1,u_0}}$ and $L$ is the operator defined by:
\begin{equation}
L =\begin{pmatrix} -\frac{r^2}{F(r)} \partial_{r_*} & -\frac{r^2}{F(r)} \\ -\frac{r^2}{F(r)}\partial_{r_*}^2- \Delta_{S^2} - \frac{2M}{r}\left(1-\frac{M}{r} \right) & -\frac{r^2}{F(r)} \partial_{r_*} \end{pmatrix} . \end{equation}
In this case we say that the solution $\phi$ peels at order $k$ at infinity.
\end{theorem}
\begin{remark}
This means that for the wave equation on the extreme Reissner-Nordström spacetime, the classes of physical initial data at $t=0$ that ensure given degrees of regularity of the rescaled field at $\scri^+$, are characterised by exactly the same regularity and decay properties as on Minkowski spacetime.
\end{remark}
\begin{remark} \label{RkOpL}
The expression of the operator $L$ in Theorem~\ref{ThmPeelingRNEInfinity} is determined from the action of
\[ \partial_R=-\frac{1}{R^2(1-MR)^2}(\partial_{r_*}+\partial_t)=-\frac{r^4}{(r-M)^2}(\partial_{r_*}+\partial_t) = -\frac{r^2}{F(r)} (\partial_{r_*}+\partial_t)\]
on the vector ${}^t\begin{pmatrix}
\phi&\partial_t\phi
\end{pmatrix}$,
\[ \partial_R \begin{pmatrix} \phi \\ \partial_t \phi \end{pmatrix} = \begin{pmatrix} -\frac{r^2}{F(r)} \partial_{r_*} & -\frac{r^2}{F(r)} \\ -\frac{r^2}{F(r)} \partial_t^2 & -\frac{r^2}{F(r)} \partial_{r_*} \end{pmatrix} \begin{pmatrix} \phi \\ \partial_t \phi \end{pmatrix}  \, ,\]
where $\phi$ satisfies Equation~\eqref{CWERNE}. Using the aforementioned equation, the term $\partial^2_t \phi$ in $\partial_R\partial_t \phi$ can be written in terms of an operator involving only spatial derivatives acting on $\phi$.
Indeed, it is readily determined from Equation~\eqref{dAlembertianRNE} that:
\begin{equation}\Box_{\hat{g}}= \frac{r^2}{F(r)}(\partial^2_t-\partial_{r_*}^2) -\Delta_{S^2}. \end{equation}
From which it can be seen that for a solution of~\eqref{CWERNE}:
\[-\frac{r^2}{F(r)} \partial_t^2\phi = -\frac{r^2}{F(r)}\partial_{r_*}^2\phi - \Delta_{S^2}\phi - \frac{2M}{r}\left(1-\frac{M}{r} \right)\phi. \]
This leads to the expression for $L$.
\end{remark}

\subsection{Peeling at the horizon for the extreme Reissner-Nordström metric}\label{PeelingAtHorizonRNE}
The peeling result now being established at infinity, we can study how it pulls back via the Couch-Torrence inversion $\Phi$. Recall from Theorem~\ref{ThmIsomCTModCT} that $\Phi$ is an isometric involution of $(\widehat{\Ext},\hat{g})$. In particular, the pullback via $\Phi$ commutes with $\Box_{\hat{g}}$, and (Corollary~\ref{ConfWEqInvarianceCT}) if $\phi$ solves~\eqref{CWERNE} then so does $\Phi^*\phi$. Moreover, whilst it preserves time-orientation, $\Phi$ reverses the overall orientation of spacetime. Hence, for any form $\alpha$ and any $4$-form $\omega$,
\begin{equation}\label{PullbackHodgeCT} \Phi^*(\star \alpha) = -\star (\Phi^*\alpha)\, , \qquad \int \Phi^*\omega=-\int \omega .\end{equation}

With this in mind, we will now translate the quantities we considered in a neighbourhood of $i^0$ to quantities in a neighbourhood of the internal infinity $i^1$ in the past of the degenerate horizon $\scrh^+$. To make this clearer,  we can work in two charts simultaneously, using ingoing coordinates $(v,R,\omega)$ for image points and outgoing coordinates $(u,R,\omega)$ for starting points. Writing $(v',R',\omega')$ for the ingoing coordinates of the image of $\Phi(p)$ where $p=(u,R,\omega)$, we have:
\[v'=u, R'= \frac{1-RM}{M}, \omega'=\omega.\]
From this we can see that $\Phi$ maps a neighbourhood $\Omega_{u_0}$ of $i^0$ to the neighbourhood of $i^1$ given by:
\[\tilde{\Omega}_{u_0} = \{t\geq 0\}\cap\{v < u_0\}.\]
Naturality of the exterior derivative combined with the fact that $\Phi$ is an isometry of $\hat{g}$, give
\begin{equation}
\Phi^*T_{ab}(\phi)=T_{ab}(\Phi^*\phi).
\end{equation}
Pullback and contractions also commute so that:
\begin{equation}
\Phi^*{J}_b=\Phi^*(K^aT_{ab}(\phi))=\Phi^*K^a\Phi^*T_{ab}(\phi)=\Phi^*K^aT_{ab}(\Phi^*\phi).
\end{equation}
Putting everything together with~\eqref{PullbackHodgeCT} it follows that if we define:
\begin{equation}
\tilde{J}(\phi)=(\Phi^*K^a)T_{ab}(\phi),
\end{equation}
with
\begin{equation} \label{PullBackMorawetzCouchTorrence}
(\Phi^*K^a) \frac{\partial}{\partial x^a} = v^2\partial_v + 2\left(1+v\frac{1-MR}{R}\right)\partial_R =: \tilde{K}^a \frac{\partial}{\partial x^a} \, ,
\end{equation}
the Couch-Torrence inversion relates energy fluxes of $J(\phi)$ to those of $\tilde{J}(\Phi^* \phi)$. More precisely, let
\[ \scrh^+_{u_0}= \scrh^+\cap\{v < u_0\}, ~ \tilde{\mathcal{H}}_{1,u_0}=\{t=0\}\cap\{v < u_0\}, ~ \tilde{\mathcal{S}}_{u_0} = \{ t \geq 0 \} \cap \{ v=u_0\},\]
and put for a given hypersurface $\mathcal{S}$
\begin{equation}
\tilde{\mathcal{E}}_{\mathcal{S}}(\Phi^*\phi) :=  \int_{\mathcal{S}} \star \tilde{J}(\Phi^*\phi) ,
\end{equation}
then:
\begin{equation}
\tilde{\mathcal{E}}_{\scrh^+_{u_0}}(\Phi^*\phi)=\mathcal{E}_{\scri^+_{u_0}}(\phi), ~  \tilde{\mathcal{E}}_{\tilde{\mathcal{H}}_{1,u_0}}(\Phi^*\phi)= \mathcal{E}_{\mathcal{H}_{1,u_0}}(\phi), ~  \tilde{\mathcal{E}}_{\tilde{\mathcal{S}}_{u_0}}(\Phi^*\phi)= \mathcal{E}_{\mathcal{S}_{u_0}}(\phi).
\end{equation}
Finally the dictionary is completed by:
\begin{equation}
\Phi^*(\partial_{R}\phi)=(\Phi^*\partial_{R})(\Phi^*\phi)=-\partial_{R'}(\Phi^*\phi).
\end{equation}
We obtain the following by applying Theorem~\ref{ThmPeelingRNEInfinity} to $\Phi^*\phi$.
\begin{theorem} \label{ThmPeelingHorizonERN}
Let $(v,R,\omega)$ be the ingoing Eddington-Finkelstein coordinates of the extreme Reissner-Nordström exterior $\mathrm{Ext}$, $v_0 \ll -1$, $k\in \mathbb{N}$ and $\phi$ a solution to the conformal wave equation~\eqref{CWERNE}.
Then $\tilde{\mathcal{E}}_{\mathscr{H}^+_{v_0}}(\partial^q_R\nabla^p_{S^2}\phi) + \tilde{\mathcal{E}}_{\tilde{\mathcal{S}}_{v_0}}(\partial^q_R\nabla^p_{S^2}\phi) < +\infty$ for all $p,q\in \mathbb{N}$, $p+q\leq k$ if and only if the initial data $(\phi_0, \phi_1)$ on the spacelike slice $\{t=0\}$ is chosen in the completion of $C^\infty_0(]-\infty,v_0]_{r_*}\times S^2)\times C^\infty_0(]-\infty,v_0]_{r_*}\times S^2)$ in the norm:
\begin{equation} \left\lVert\begin{pmatrix}\phi_0\\\phi_1\end{pmatrix}\right\rVert_k^2=\sum_{p+q\leq k} \tilde{\mathcal{E}}_{\tilde{\mathcal{H}}_{1,v_0}}\left(\tilde{L}^q\nabla^p_{S^2}\begin{pmatrix}\phi_0\\\phi_1\end{pmatrix} \right), \end{equation}
where we replace $\phi$ by $\phi_0$ and $\partial_t\phi$ by $\phi_1$ in the expression of
the energy flux $\tilde{\mathcal{E}}_{\tilde{\mathcal{H}}_{1,v_0}}$ and $\tilde{L}$ is the operator defined by:
\begin{equation} \label{OpLTilde}
\tilde{L} =\begin{pmatrix} -\frac{r^2}{F(r)} \partial_{r_*} & \frac{r^2}{F(r)} \\ \frac{r^2}{F(r)}\partial_{r_*}^2+ \Delta_{S^2} + \frac{2M}{r}\left(1-\frac{M}{r} \right) & -\frac{r^2}{F(r)} \partial_{r_*} \end{pmatrix} . \end{equation}
In this case we say that the solution $\phi$ peels at order $k$ at the degenerate horizon $\scrh^+$.
\end{theorem}
\begin{remark}
Note that, due to the spherical symmetry, one can control the purely radial regularity at the horizon as shown in Theorem 4 in \cite{MaNi2009}.
\end{remark}
\begin{remark}
The operator $\tilde{L}$ is obtained exactly as $L$ in Remark \ref{RkOpL}, but in this case we work with ingoing coordinates $(v,R,\omega)$, whence
\[ \partial_R = \frac{r^2}{F(r)} \left( \partial_t - \partial_{r_*} \right) \, .\]
Consequently
\[ \partial_R \begin{pmatrix} \phi \\ \partial_t \phi \end{pmatrix} = \begin{pmatrix} -\frac{r^2}{F(r)} \partial_{r_*} & \frac{r^2}{F(r)} \\ \frac{r^2}{F(r)} \partial_t^2 & -\frac{r^2}{F(r)} \partial_{r_*} \end{pmatrix} \begin{pmatrix} \phi \\ \partial_t \phi \end{pmatrix}  \]
and if $\phi$ satisfies Equation \eqref{CWERNE} then the action of $\partial_R$ on the vector ${}^t\begin{pmatrix} \phi & \partial_t \phi \end{pmatrix}$ is given by \eqref{OpLTilde}.
\end{remark}
\subsection{Peeling for massive fields at the horizon of the extreme Reissner-Nordström metric}\label{PeelingKleinGordonRNE}
The peeling established directly at the horizon is proved for the rescaled metric $\hat{g} = R^2 g$ but contrary to what happens at infinity, $R \rightarrow 1/M >0$ at the horizon, hence the conformal rescaling is not singular there. This means that we can deal with massive fields. Indeed, using the transformation law for the d'Alembertian under a conformal rescaling
\[ \square_g + \frac16 \mathrm{Scal}_g = \Omega^3 \left( \square_{\hat{g}} + \frac16 \mathrm{Scal}_{\hat{g}} \right) \Omega^{-1} \, ,\]
we obtain that for any $m\geq0$ ,
\[ \square_g  + \frac16 \mathrm{Scal}_g  + m^2 = \Omega^3 \left( \square_{\hat{g}} + \frac16 \mathrm{Scal}_{\hat{g}} + \Omega^{-2} m^2 \right) \Omega^{-1} \, . \]
Hence, for any distribution $\psi$ on the exterior of the black hole, putting $\phi = \Omega^{-1} \psi$,
\[ \left( \square_g  + \frac16 \mathrm{Scal}_g  + m^2 \right) \psi = \Omega^3 \left( \square_{\hat{g}} + \frac16 \mathrm{Scal}_{\hat{g}} + \Omega^{-2} m^2 \right) \phi \, .\]
\begin{remark}
In the conformal density formalism of Section~\ref{IntroConformaldAlembertian}, this amounts to studying the equation:
\[ \Box_c \underline{\phi} +m^2\sigma_g^{-2}\underline{\phi}=0, \quad \underline{\phi}\in\mathcal{E}[-1]. \]
Where $g$ is the \emph{physical} metric and $\sigma_g$ the canonical conformal $1$-density it determines. Since if $\hat{g}=\Omega^2 g$, 
$\sigma_{\hat{g}}=\Omega^{-1}\sigma_g$ when writing this in a scale this leads to the above equations.
\end{remark}
With the conformal factor $\Omega = R$, $\psi$ is a solution to
\begin{equation} \label{WEqPo}
\left( \square_g  + \frac16 \mathrm{Scal}_g  + m^2 \right) \psi =0
\end{equation}
if and only if $\phi = R^{-1} \psi$ satisfies
\begin{equation} \label{RescWEqPot}
\left( \square_{\hat{g}} + \frac16 \mathrm{Scal}_{\hat{g}} + R^{-2} m^2 \right) \phi =0 \, .
\end{equation}
Moreover, $\Phi^* \phi$ being a solution to \eqref{RescWEqPot} is equivalent to $ \phi$ satisfying the equation
\begin{equation} \label{RescWeqPotCT}
\left( \square_{\hat{g}} + \frac16 \mathrm{Scal}_{\hat{g}} + \left( \frac{Mm}{1-MR} \right)^{2} \right) \phi =0 \, ,
\end{equation}
since under the Couch-Torrence inversion, $1/R=r$ is transformed into
\[ \frac{rM}{r-M} = \frac{M}{1-MR} \, .\]
We can study the peeling at infinity and then translate the results, using the Couch-Torrence inversion as a dictionary, at the horizon. The coefficient $(Mm/(1-MR))^2$ tends to $(Mm)^2$ as $r\rightarrow \infty$ and is therefore bounded in the neighbourhood of infinity, however it does not decay at infinity. Using the same stress-energy tensor and observer as for the conformal wave equation to define the energy current $J$, we obtain the following approximate conservation law for $\phi$ solution to  \eqref{RescWeqPotCT}
\begin{align*}
\hat\nabla^a J_a =&\, \nabla^{(a} K^{b)} T_{ab} - \frac{1}{6} \mathrm{Scal}_{\hat{g}} \phi \nabla_{K} \phi - m^2 \phi \nabla_{K} \phi \\
=&  -2 \, {\left(R^{2} M^{2} - R M\right)} u^{2} \phi \frac{\partial\,\phi}{\partial u} + 4 \, {\left(R^{2} M^{2} - R M + {\left(R^{3} M^{2} - R^{2} M\right)} u\right)} \phi \frac{\partial\,\phi}{\partial R} \\
& - 2 \, {\left(2 \, R^{3} M^{2} - 3 \, R^{2} M + {\left(R^{4} M^{2} - R^{3} M\right)} u\right)} \left( \frac{\partial\,\phi}{\partial R}\right)^{2} \\
& - \left( \frac{Mm}{1-MR} \right)^{2} u^2 \phi \frac{\partial \phi}{\partial u} + \left( \frac{Mm}{1-MR} \right)^{2} 2 (1+uR) \phi \frac{\partial \phi}{\partial R} \, .
\end{align*}
The third line in the equation above has no decay and is therefore harder to control than the other two, however, the factor $1/\vert u \vert$ that we gain in the splitting of the $4$-volume measure allows to control this term exactly as we did for the first order estimate in the massless case. The rest of the proof goes through without modifications. We obtain the following result:
\begin{theorem} \label{PeelingHorizonRNE}
Theorems \ref{ThmPeelingRNEInfinity} and \ref{ThmPeelingHorizonERN} are valid respectively for solutions to Equations \eqref{RescWeqPotCT} and \eqref{RescWEqPot} with the operators $L$ and $\tilde{L}$ now given by
\begin{eqnarray*}
L &=&\begin{pmatrix} -\frac{r^2}{F(r)} \partial_{r_*} & -\frac{r^2}{F(r)} \\ -\frac{r^2}{F(r)}\partial_{r_*}^2- \Delta_{S^2} - \frac{2M}{r}\left(1-\frac{M}{r} \right) + \left( \frac{Mmr}{r-M}\right)^2& -\frac{r^2}{F(r)} \partial_{r_*} \end{pmatrix} , \\
\tilde{L} &=&\begin{pmatrix} -\frac{r^2}{F(r)} \partial_{r_*} & \frac{r^2}{F(r)} \\ \frac{r^2}{F(r)}\partial_{r_*}^2+ \Delta_{S^2} + \frac{2M}{r}\left(1-\frac{M}{r} \right) -r^2m^2& -\frac{r^2}{F(r)} \partial_{r_*} \end{pmatrix} .
\end{eqnarray*}
\end{theorem}

\section{The Couch-Torrence inversion for a spherically symmetric degenerate horizon and peeling} \label{LooseEndCT}

The general form of the metric near a future degenerate horizon $\mathscr{H}^+$can be found in \cite{KuLu2013}, expressed in Gaussian null coordinates, with a horizon located at $r=M$
\begin{equation} \label{DegHorMet}
g = -2 \d v \left( \d r + (r-M) h_a (r,x) \d x^a - \frac12 (r-M)^2 \Gamma(r,x) \d v \right) + \gamma_{ab} (r,x) \d x^a \d x^b \, .
\end{equation}
In this case, $M$ has no physical significance other than the value of a coordinate. We study a simple case for which the metric is spherically symmetric:
\begin{equation} \label{SphSymDegHorMet}
g = -2 \d v \left( \d r - \frac12 (r-M)^2 \Gamma(r) \d v \right) - r^2 \d \omega^2 \, ,
\end{equation}
where the function $\Gamma$ is positive, analytic and does not vanish at $r=M$. In the case of extreme Reissner-Nordström, we have
\[ \Gamma(r) = \frac{1}{r^2} \rightarrow \frac{1}{M^2} \mbox{ at the horizon.}\]

\subsection{The image asymptotic region and its conformal compactification}

Let us introduce a manifold, to be made more precise later, that we shall denote by $\mathscr{S}$ with coordinates $(u,z,\omega)$ ($\omega \in S^2$) and define a smooth map $\Phi$ from the $r>M$ region near the degenerate horizon into $\mathscr{S}$ that mimics the Couch-Torrence inversion:
\[ \Phi : (v,r,\omega) \mapsto (v, \frac{rM}{r-M},\omega)\in \mathscr{S}.\]
We assume that we are working on open sets of the two manifolds such that this is a diffeomorphism with inverse:
\[ \Phi^{-1}: (u,z,\omega) \mapsto (u, \frac{zM}{z-M},\omega), \]
Pulling back $g$ to $\mathscr{S}$ by $\Phi^{-1}$, we find:
\begin{eqnarray}
(\Phi^{-1})^*g &=& -2 \d u \left( -\left( \frac{M}{z-M} \right)^2 \d z - \frac12 \frac{M^4}{(z-M)^2} \Gamma\left(\frac{zM}{z-M}\right) \d u \right) - \left( \frac{Mz}{z-M} \right)^2 \d \omega^2 \nonumber \\
&=& \left( \frac{M}{z-M} \right)^2 \left( 2 \d u \d z + M^2 \Gamma\left(\frac{zM}{z-M}\right) \d u^2 - z^2 \d \omega^2 \right) \, . \label{GenEBHCTMet}
\end{eqnarray}

%
This shows that we can generalise the Couch-Torrence inversion to a conformal isometry between a neighbourhood of the $r>M$ region near a degenerate horizon and a manifold $\mathscr{S}$ with metric
\[ h=2 \d u \d z + M^2 \Gamma\left(\frac{zM}{z-M}\right) \d u^2 - z^2 \d \omega^2.\]

Unless we are in the extreme Reissner-Nordström framework, we cannot assume $\Gamma(M)=M^{-2}$ and it looks like the metric \eqref{GenEBHCTMet} has peculiar asymptotic behaviour. To clarify this issue, let us do a detailed calculation using Schwarzschild-type coordinates. We define the variables $r_*$ and $t$ by
\begin{equation} \label{rstartSphSymDegHor}
(r-M)^2\Gamma(r) \d r_* = \d r  \, ,~ t = v-r_* \, .
\end{equation}
The metric \eqref{SphSymDegHorMet} is expressed as
\[ g= (r-M)^2 \Gamma(r) \d t^2 - \frac{1}{(r-M)^2\Gamma(r)} \d r^2 -r^2 \d \omega^2\, .\]
We perform the Couch-Torrence inversion . Let us also introduce coordinates $\tau$ and $z_*$ on $\mathscr{S}$ defined by:
\begin{equation} \label{DefTau} \tau=u+z_*, \quad z_*(z)= -r_*\left(\frac{Mz}{z-M}\right). \end{equation}
Note that by definition we have:
\begin{equation} \label{DefzstarDiff} \Gamma\left(\frac{Mz}{z-M}\right)M^2 \d z_* = \d z \end{equation}
In terms of variables $(t,r,\omega)$ and $(\tau, z, \omega)$, our map can be written:
\[ \Phi: (t,r,\omega)\mapsto (t, \frac{rM}{r-M}, \omega) \in \mathscr{S}, \]
 and inverse:
 \[ \Phi^{-1}: (\tau,z,\omega)\mapsto (\tau, \frac{zM}{z-M},\omega). \]
 Now pulling back $g$ to $\mathscr{S}$ by $\Phi^{-1}$ gives:
\begin{align*}
(\Phi^{-1})^*g &=  \frac{M^4}{(z-M)^2}\Gamma\left(\frac{zM}{z-M}\right) \d \tau^2 - \frac{1}{(z-M)^2\Gamma\left(\frac{zM}{z-M}\right)} \d z^2 - \left( \frac{Mz}{z-M} \right)^2 \d \omega^2 \, , \\
&= \left( \frac{M}{z-M} \right)^2 \left( M^2 \Gamma\left(\frac{zM}{z-M}\right) \d \tau^2 - \frac{1}{M^2 \Gamma\left(\frac{zM}{z-M}\right)} \d z^2 - z^2 \d \omega^2 \right) \, .
\end{align*}
Let $\alpha = M\sqrt{\Gamma(M)}$,
\[ (\Phi^{-1})^*g = \left( \frac{M}{(z-M) \alpha}\right)^2 \left( \alpha^2 M^2 \Gamma\left(\frac{zM}{z-M}\right) \d \tau^2 - \frac{\alpha^2}{M^2 \Gamma(\frac{zM}{z-M})} \d z^2 - \alpha^2 z^2 \d \omega^2 \right) \]
and putting
\[ \tau = \frac{\tilde{\tau}}{\alpha^2 } \, ,~ \Psi \left( \frac{1}{z} \right) = \frac{M^2\Gamma\left(\frac{zM}{z-M}\right)}{\alpha^2} \, ,\]
we get
\[ (\Phi^{-1})^*g = \left( \frac{M}{(z-M) \alpha}\right)^2 \left( \Psi \d {\tilde{\tau}\,}^2 - \frac{1}{\Psi} \d z^2 - \alpha^2 z^2 \d \omega^2 \right) \, ,\]
where $\Psi$ is an analytic positive function such that $\Psi (0)=1$.

The metric $(\Phi^{-1})^*g$ is conformally equivalent to
\[ \tilde{g} = \left( \frac{(z-M) \alpha}{M}\right)^2 (\Phi^{-1})^*g = \Psi \d {\tilde{\tau}\,}^2 - \frac{1}{\Psi} \d z^2 - \alpha^2 z^2 \d \omega^2\]
whose asymptotic structure as $z\rightarrow +\infty$ has the form
\[ \tilde{g}_\infty = \d \tilde{\tau}^2 - \d z^2 - \alpha^2 z^2 \d \omega^2 \, .\]
In the extreme Reissner-Nordström case, $\alpha =1$ and $\tilde{g}_\infty$ is the Minkowski metric. In general, we cannot expect that $\alpha =1$ and the metric $\tilde{g}$ will be asymptotically ``conical''. This is still asymptotically flat, but with a different rate of fall-off at infinity of the curvature on the spheres, compared to the Minkowski or the Schwarzschild metric. Near spacelike infinity, the spacelike slices look like the large ends of cones rather than the neighbourhood of infinity on $\R^3$.

We now turn to the conformal compactification of \eqref{GenEBHCTMet} near $z=+\infty$. Instead of multiplying the metric $(\Phi^{-1})^*g$ by $Z^2$, where $Z = 1/z$, we pull back via the Couch-Torrence inversion the rescaled metric
\[ \hat{g} = R^2 g \, ,~ R = 1/r \, ,\]
i.e.
\begin{equation} \label{SphSymDegHorMetResc}
\hat{g} = (1-MR)^2 \Gamma \left( \frac{1}{R} \right) \d v^2 +2 \d v \d R - \d \omega^2 \, .
\end{equation}
We have, using~\eqref{GenEBHCTMet} and returning to our initial coordinate system on $\mathscr{S}$:
\begin{align*}
(\Phi^{-1})^*\hat{g} &=  (R\circ \Phi^{-1})(\Phi^{-1})^*g = \frac{(z-M)^2}{M^2z^2} (\Phi^{-1})^*g \\
&= \frac{1}{z^2} \left( 2 \d u \d z + M^2 \Gamma\left(\frac{zM}{z-M}\right) \d u^2 - z^2 \d \omega^2 \right)
\end{align*}
putting
\[ Z = \frac{1}{z} \, ,~  f(Z)= M^2 \Gamma\left(\frac{zM}{z-M}\right) \, ,\]
we arrive at
\begin{equation}\label{RescaledMetricAtInfinityS} (\Phi^{-1})^*\hat{g} = Z^2f(Z)\d u^2 -2\d u\d Z - \d \omega^2 \, .\end{equation}
%
%
%
%
%
%
%
%
%
The function $f$ is analytic and positive on an interval of the form $[0,a [$, $a>0$ and the case $f(0)=1$ corresponds to those situations where $\Gamma(M) = 1/M^2$. In the extreme Reissner-Nordström case, we have $f(Z) = (1-MZ)^2$.

\subsection{Peeling for the wave equation on the image asymptotic region}

In this section, we revert to more usual notations, making the replacements $z\rightarrow r$, $Z\rightarrow R$ for the radial and inverted radial variables respectively. We will study the peeling at null infinity for the metric
\begin{equation} \label{ImageEBHCTMetric}
\hat{g} = R^2 f (R) \d u^2 - 2 \d u \d R - \d \omega^2 \, .
\end{equation}
Future null infinity ($\scri^+$) is described as $\{ R=0 \} \times \R_u \times S^2_\omega$. The scalar curvature of $\hat{g}$ is
\[ \mathrm{Scal}_{\hat{g}} = -R^{2} f''(R) - 4 \, R f'(R) + 2(1-f(R)) \]
and it is nonzero at $\scri^+$ unless $f(0)=1$.

Recall from equations~\eqref{DefTau} and \eqref{DefzstarDiff} that the variables $t$ and $r_*$ (denoted by $\tau$ and $z_*$ previously) satisfy:
\begin{equation}\label{DefRstar2} f(R) \d r_* = \d r \mbox{ or equivalently } - R^2 f(R) \d r_* = \d R \mbox{ and } t = u+r_* \, .\end{equation}
The metric $\hat{g}$ in terms of $(t,r_*,\omega)$ is expressed as
\begin{equation} \label{ImageEBHCTMetrictr*}
\hat{g} = \frac{\tilde{f}(r)}{r^2} \left( \d t^2 - \d r_*^2 \right) - \d \omega^2 \, ,~ \tilde{f}(r) := f\left( \frac{1}{r} \right) \, .
\end{equation}
The d'Alembertian for $\hat{g}$ is given by
\begin{equation}
\square_{\hat{g}}= -R^{2} f\left(R\right) \frac{\partial^2}{\partial R ^ 2} - {\left(R^{2} f'(R) + 2 \, R f\left(R\right)\right)}  \frac{\partial}{\partial R} - 2 \,  \frac{\partial^2}{\partial u\partial R} -\Delta_{S^2}.
\end{equation}
In terms of variables $(t,r_*,\omega)$ it takes the form
\begin{equation}
\square_{\hat{g}} = \frac{r^{2}}{\tilde{f}(r)} \left( \frac{\partial^2}{\partial t^ 2} - \frac{\partial^2}{\partial r_*^ 2} \right) -\Delta_{S^2}
\end{equation}
so the conformal wave equation reads
\begin{equation} \label{CWESph}
\frac{r^{2}}{\tilde{f}(r)} \left( \frac{\partial^2\phi}{\partial t^ 2} - \frac{\partial^2\phi}{\partial r_*^ 2} \right) -\Delta_{S^2} \phi + \frac16 \mathrm{Scal}_{\hat{g}} \phi =0 \, .
\end{equation}
We work in a neighbourhood of spacelike infinity of the form
\[ \Omega_{u_0} = \{ t\geq 0\} \cap \{ u< u_0 \} \]
for $u_0 \ll-1$ and we use the same foliation
\[ \mathcal{H}_{s,u_0} = \{ u = -s r_* \, ,~u< u_0 \} \, ,~ 0 \leq s \leq 1 \]
as in the extreme Reissner-Nordström case. The analogue of Lemma~\ref{ApproxCloseI0} in our new framework is the following.
\begin{lemma}\label{RuBounded2}
Let $0<\varepsilon<1$, then one can find $u_0<0$, $\vert u_0\vert $ large enough, such that in $\Omega_{u_0}$,
\[\frac{1-\varepsilon}{f(0)} < Rr_* < \frac{1+\varepsilon}{f(0)}, \quad 0 < R\vert u\vert  < \frac{1+\varepsilon}{f(0)}, \quad f(0)(1-\varepsilon) < f(R)< f(0)<(1+\varepsilon). \]
\end{lemma}
\begin{proof}
It is a direct consequence of L'Hôpital's rule and \eqref{DefRstar2}.
\end{proof}

We choose once more the family of observers associated to the Morawetz vector field $K$ defined in \eqref{Morawetz}. Its Killing form for \eqref{ImageEBHCTMetric} reads
\begin{equation} \label{MorawetzKillingFormEBHCT}
\mathcal{L}_{K} \hat{g} = 2R \left(  2 (1- f(R)) - R (1+ uR) f'(R) \right) \mathrm{d} u\otimes \mathrm{d} u \, .
\end{equation}
In the extreme Reissner-Nordström case, the Killing form \eqref{MorawetzKillingFormRNE} of $K$ vanished at order $2$ at $\scri^+$, whereas \eqref{MorawetzKillingFormEBHCT} only generically vanishes at order $1$ (again we need $f(0)=1$ for it to vanish at order $2$). Moreover, the Morawetz vector field was timelike in the neighbourhood of $i^0$. Here, we have the following restriction:

\begin{lemma}\label{Restrictionf(0)}
We can choose $u_0 \ll-1$ such that the Morawetz vector field is uniformly timelike on $\Omega_{u_0}$ if and only if $f(0)>\frac{3}{4}$.
\end{lemma}
\begin{proof}
The \enquote{squared norm} of $K$ is given by:
\begin{equation} \hat{g}(K,K)=R^2f(R)u^4+4u^2(1+Ru)=4u^2\left(1+ Ru + \frac{(R^2u^2)f(R)}{4} \right). \end{equation}
On $\mathcal{H}_{1,u_0}$, as one approaches $i_0$, we have $R\to 0$, $f(R)\to f(0)$ and $Ru\rightarrow -\frac{1}{f(0)}$ whence
\[1+ Ru +\frac{1}{4}(Ru)^2f(R) \to 1 -\frac{3}{4f(0)}  \]
which is negative if $f(0)<\frac{3}{4}$. In this case, $K^a$ is not timelike on the whole of $\Omega_{u_0}$. On $\mathcal{H}_{0,u_0}$, $\hat{g}(K,K)=4u^2$ independently of the value of $f(0)$ and on $\mathcal{H}_{s,u_0}$, $0<s\leq 1$ we have:
\[\begin{aligned} \hat{g}(K,K)&=4u^2\left( 1-sRr_*+\frac{1}{4}s^2(Rr_*)^2f(R)\right) \\&\geq 4u^2\left(1-s\frac{1+\varepsilon}{f(0)}+\frac{1}{4}s^2\frac{(1-\varepsilon)^3f(0)}{f(0)^2} \right) \\&= 4u^2\left(1-\frac{s(4-s)}{4f(0)} + O(\varepsilon) \right).\end{aligned} \]
On $[0,1]$, $s(4-s)$ is increasing and varies between $0$ and $3$. So if $f(0)>\frac{3}{4}$, choosing $\varepsilon$ sufficiently small enables us to ensure that $K$ is uniformly timelike on the whole of $\Omega_{u_0}$. The case $f(0)=\frac{3}{4}$ is marginal, $K$ becomes null at $i_0$ and, depending on the behaviour of $r_*$, may also become null or timelike near $\mathcal{H}_{1,u_0}$.
\end{proof}

Note that the closer $f(0)$ is to $\frac{3}{4}$, the smaller one has to choose $\Omega_{u_0}$ in order to ensure that $K$ remains timelike uniformly on $\Omega_{u_0}$.

With this restriction in mind, we turn to the calculation of the energy flux across $\mathcal{H}_{s,u_0}$.
 A normal vector field to $\mathcal{H}_{s,u_0}$ is given by
\[ n = \partial_u + R^2 f(R) \frac{r_*}{u}\left( 1-s \right) \partial_R = \partial_u + R^2 f(R) \frac{s-1}{s} \partial_R \]
and
\[ l = - \partial_R \]
is a future-oriented transverse vector field to $\mathcal{H}_{s,u_0}$ such that $\hat{g} (l,n) =1$.
The energy on $\mathcal{H}_{s,u_0}$ is:
\begin{eqnarray*}
\mathcal{E}_{\mathcal{H}_{s,u_0}}(\phi)&=& \int_{]-\infty , u_0[_u \times S^2_\omega} \left( R^{2} u^{2} f\left(R\right) \frac{\partial\,\phi}{\partial u} \frac{\partial\,\phi}{\partial R} +  u^{2} \frac{\partial\,\phi}{\partial u}^{2} \right. \\
&&+ {\left(\frac{R^{4} u^{2} f\left(R\right)^{2} }{2s} + R^3 f(R)  (2-s)\frac{u}{s} + {\frac{R^{2} f(R) (2-s)}{s}} \right)} \frac{\partial\,\phi}{\partial R}^{2} \\
&&\left. + {\left(\frac{R^{2} u^{2} f\left(R\right) }{2s}+ R  u + 1 \right)} \left\vert \nabla_{S^2} \phi \right\vert^2  \right) \d u \d^2 \omega \, ,\\
&=&\int_{]-\infty , u_0[_u \times S^2_\omega} \left( R^{2} u^{2} f\left(R\right) \frac{\partial\,\phi}{\partial u} \frac{\partial\,\phi}{\partial R} +  u^{2} \frac{\partial\,\phi}{\partial u}^{2} \right. \\
&&+ \frac{Rr_* f(R)}{2} \frac{R}{\vert u \vert} {\left((R \vert u \vert)^2 f(R) - 2(2-s) R\vert u \vert +  2(2-s)\right)} \frac{\partial\,\phi}{\partial R}^{2} \\
&&\left. + {\left( 1 + R \vert u \vert \left( \frac{Rr_* f(R) }{2}-1\right) \right)} \left\vert \nabla_{S^2} \phi \right\vert^2  \right) \d u \d^2 \omega \, .
\end{eqnarray*}
On $\scri^+ \cap \{ u<u_0\}$, this reduces to
\[ \mathcal{E}_{\mathcal{H}_{0,u_0}}(\phi) =  \int_{]-\infty , u_0[_u \times S^2_\omega} \left( u^{2} \frac{\partial\,\phi}{\partial u}^{2} + \left\vert \nabla_{S^2} \phi \right\vert^2  \right) \d u \d^2 \omega \, . \]
\begin{proposition} \label{EqEnAF}
The energy on $\mathcal{H}_{s,u_0}$ is equivalent, uniformly in $s \in [0,1]$ to the simpler expression
\[ \int_{]-\infty , u_0[_u \times S^2_\omega} \left( u^{2} \frac{\partial\,\phi}{\partial u}^{2} + \frac{R}{\vert u \vert} \frac{\partial\,\phi}{\partial R}^{2} + \left\vert \nabla_{S^2} \phi \right\vert^2  \right) \d u \d^2 \omega \, . \]
\end{proposition}
\begin{proof}
Using Lemma~\ref{RuBounded2}, we have
\[\begin{aligned} 1+ R\vert u\vert \left(\frac{Rr_*f(R)}{2}-1 \right)&\geq 1+ R\vert u\vert \left( \frac{(1-\varepsilon)^2}{2}-1 \right)\\
&\geq 1+ \frac{(1+\varepsilon)}{f(0)}\left(\frac{(1-\varepsilon)^2}{2}-1 \right)\\&\geq 1-\frac{1}{2f(0)} + O(\varepsilon)> \frac{1}{3}+ O(\varepsilon). \end{aligned} \]
In the last inequality we use the restriction identified in Lemma~\ref{Restrictionf(0)}.

As for the factor of $(\partial_R \phi)^2$, it is the product of
\[\frac{Rr_*f(R)}{2}\frac{R}{\vert u\vert }= \left(\frac{1}{2} + O(\varepsilon)\right)\frac{R}{\vert u\vert } \]
and the quantity:
\[ P(R\vert u\vert ):=(R\vert u\vert )^2f(R)-2(2-s)R\vert u\vert +2(2-s).\]
On $\mathcal{H}_{s,u_0}$ for $0<s\leq 1$:
\[\begin{aligned} P(R\vert u\vert )&=(Rr_*)^2f(R)s^2+2(2-s)(1-(Rr_*)s)
\\&\geq \frac{(1-\varepsilon)^3}{f(0)}s^2 +2(2-s)\left(1-\frac{1+\varepsilon}{f(0)}s \right)\\
&= \frac{1}{f(0)}\left(s^2 + 2(2-s)f(0)-2(2-s)s \right)+ O(\varepsilon)
\\&= \frac{1}{f(0)}\left(3s^2 -2s(f(0)+2)+4f(0)  \right)+O(\varepsilon).\end{aligned}\]

Let us briefly consider the polynomial:
\[ Q(s):=3s^2 - 2s(f(0)+2)+4f(0), \]
$Q$ attains its minimum value at $s_\textrm{min}=\frac{1}{3}(f(0)+2)$, this is in the interval $[0,1]$ if and only if $f(0)\leq 1$ (we assume $f(0)>0$).
When this is the case then the minimum value is given by:
\[ Q_{\textrm{min}}=4f(0)-\frac{1}{3}(f(0)+2)^2, \]
which is positive for $0<f(0)\leq 1$ if and only if $f(0)>4-2\sqrt{3}$. Since $4-2\sqrt{3} < \frac{3}{4}$ and we assume $f(0)>\frac{3}{4}$, it follows that: $Q_\textrm{min}> 0$.
Note that, in this case, $Q_\textrm{min}>\frac{23}{48}$.

When $f(0)>1$, the minimum is reached beyond $s=1$, thus the minimum on the interval $[0,1]$ is:
\[Q(1)=2f(0) -1>1. \]
It follows that, overall when $f(0)>\frac{3}{4}$:
\[ P(R\vert u\vert ) \geq \frac{23}{48 f(0)} + O(\varepsilon)>\frac{23}{36} +O(\varepsilon),\]
and:
\[ \frac{Rr_*f(R)}{2}\frac{R}{\vert u\vert }P(R\vert u\vert ) \geq \left( \frac{23}{72} + O(\varepsilon) \right)\frac{R}{\vert u\vert },\]

Finally we turn to the control of the term involving $\partial_R \phi \partial_u \phi$ on $\Omega_{u_0}$,
\[\begin{aligned} \left\lvert R^2u^2f(R)\partial_u\phi \partial_R\phi \right\rvert&=\sqrt{\frac{R}{\vert u\vert }}\vert u\vert (R\vert u\vert )^\frac{3}{2}f(R)\left\lvert\partial_u\phi \partial_R\phi\right\rvert \\ &\leq \sqrt{\frac{R}{\vert u\vert }}\vert u\vert \frac{(1+\varepsilon)^{\frac{5}{2}}}{\sqrt{f(0)}}\left\lvert\partial_u\phi \partial_R\phi\right\rvert.\end{aligned}\]
Now for any $\lambda >0$ we have:
\[ \left\lvert R^2u^2f(R)\partial_u\phi \partial_R\phi \right\rvert \leq (1+\varepsilon)^{\frac{5}{2}}\left( \frac{\lambda^2}{2}u^2(\partial_u\phi)^2 +\frac{1}{2\lambda^2f(0)}\frac{R}{\vert u\vert }(\partial_R\phi)^2 \right). \]
The proposition will be proved if we can choose $\lambda$ such that the following conditions hold:
\[ \frac{36}{23}\frac{1}{f(0)}<\lambda^2<2.\]
Since $f(0)>\frac{3}{4}$ it is sufficient to impose:
\[ \frac{104}{69}=1+\frac{35}{69} < \lambda^2 < 2, \]
which is always possible.
\end{proof}
The error terms in the approximate conservation law for $J^a$ are
\begin{eqnarray}
\hat{\nabla}^a J_a &=& \frac{1}{6} \, {\left(R^{2} \frac{\partial^2\,f}{\partial R ^ 2} + 4 \, R f'(R) + 2 \, f\left(R\right) - 2\right)} u^{2} \phi \frac{\partial\,\phi}{\partial u} \nonumber\\
&&- \frac{1}{3} \, \bigg(R^{2} f''( R) + {\left(R^{3} f''(R) + 4 \, R^{2} f'(R) + 2 \, R f\left(R\right) - 2 \, R\right)} u \nonumber \\
&& \hspace{0.5in} + 4 \, R f'(R) + 2 \, f\left(R\right) - 2\bigg) \phi \frac{\partial\,\phi}{\partial R} \nonumber \\
&&- 2 \, {\left(R^{2} (1+R u) f'(R)  + 2 \, R f\left(R\right) - 2 \, R\right)} \frac{\partial\,\phi}{\partial R}^{2} \, .
\end{eqnarray}
Thanks to the result of Proposition \ref{EqEnAF}, these can be dealt with as before; the same is true for higher order estimates. We obtain the following theorem.
\begin{theorem}
For the metric $\hat{g}$ defined on $\mathscr{S}$ by \eqref{ImageEBHCTMetric} and under the assumption that $f(0)>3/4$, Theorem \ref{ThmPeelingRNEInfinity} is valid for solutions to Equation \eqref{CWESph} with the operator $L$ given by
\[ L =\begin{pmatrix} -\frac{r^2}{f(1/r)} \partial_{r_*} & -\frac{r^2}{f(1/r)} \\ -\frac{r^2}{f(1/r)}\partial_{r_*}^2- \Delta_{S^2} +\frac16 \mathrm{Scal}_{\hat{g}} & -\frac{r^2}{f(1/r)} \partial_{r_*} \end{pmatrix} . \]
\end{theorem}

\subsection{Peeling at the degenerate horizon}
As before, the Couch-Torrence inversion provides a dictionary between objects near the horizon with rescaled metric defined by~\eqref{SphSymDegHorMetResc} and those in $\mathscr{S}$ with the metric~\eqref{RescaledMetricAtInfinityS}. By construction of $\mathscr{S}$, it is in all points identical to the dictionary in Section~\ref{PeelingAtHorizonRNE}, under the proviso that we relate objects on two different spacetimes. It is again possible to include massive fields in our treatment following the same steps outlined in Section~\ref{PeelingKleinGordonRNE}. This leads to the following generalisation of Theorem~\ref{PeelingHorizonRNE}:

\begin{theorem}
The conclusions of Theorem~\ref{ThmPeelingHorizonERN} hold for solutions of the equation:
\[ \Box_{\hat{g}}\phi +\frac{1}{6}\textrm{Scal}_{\hat{g}} \phi +r^2m^2 \phi =0 \]
where $\hat{g}$ is the metric defined by Equation~\eqref{SphSymDegHorMetResc} of which the scalar curvature is given by:
\[\textrm{Scal}_{\hat{g}}=2 - \frac{\d^2}{\d R^2}\left(\Gamma\left(\frac{1}{R}\right)(1-RM)^2\right)=2-r^2\frac{\d}{\d r}r^2\frac{\d}{\d r}\left(\frac{\Gamma(r)(r-m)^2}{r^2} \right) \]
and the operator $\tilde{L}$ is replaced by:
\[ \tilde{L}=\frac{r^2}{(r-M)^2\Gamma(r)} \begin{pmatrix} -\partial_{r_*} & 1 \\ \partial_{r_*}^2 +\frac{(r-M)^2\Gamma(r)}{r^2}\Delta_{S^2}-\frac{(r-M)^2\Gamma(r)}{6r^2}\textrm{Scal}_{\hat{g}} -r^2m^2 & -\partial_{r_*} \end{pmatrix}. \]
\end{theorem}

\section{Conclusion}
The peeling at an extreme horizon is analogous (modified merely by a finite conformal transformation) to the peeling at a {\bf conformally rescaled} asymptotically flat infinity with conformal factor $1/r$. The analogy of behaviour between a degenerate horizon and null infinity for an asymptotically flat spacetime should therefore be read between the physical field at the horizon (or its finite conformal rescaling $r\psi$) and a conformally rescaled field $r \psi$ at null infinity. Also, an important feature of the peeling at a degenerate horizon is that it is also valid for massive fields; in fact lower order perturbations of the d'Alembertian with coefficients bounded at all orders in the neighbourhood of internal infinity can also be accommodated.

\vspace{0.5in}

\noindent{\bf Acknowledgements.} JB gratefully acknowledges that part of this work was supported by the French \enquote{Investissements d’Avenir} program, project ISITE-BFC (contract ANR- 15-IDEX-0003). EG acknowledges that part of this work has been funded by
l'Agence Nationale de la Recherche, project StronG ANR-22-CE31-0015-01. JPN would like to thank Atul Sharma and Lionel Mason for stimulating discussions. He is also grateful to The Mathematical Institute, Oxford, for hospitality while this work was being developed and thanks the University of Brest and the LMBA for financial support during that time.


\begin{thebibliography}{100}
%
\bibitem{Are2018Book} S. Aretakis, {\em Dynamics of degenerate horizons}, Springer Briefs in Mathematical Physics 33, 2018.
%
\bibitem{BiFri2013} P. Bizon, H. Friedrich, {\em A remark about wave equations on the extreme Reissner-Nordström black hole exterior}, Class. Quantum Grav. {\bf 30} (2013), 065001.
%
\bibitem{JB2020} J. Borthwick, {\em Scattering theory for Dirac fields near an extreme Kerr-de Sitter black hole}, Ann. Inst. Fourier, in press; arXiv:2005.01036.

\bibitem{CoTo1983} W.E. Couch, R.J. Torrence, {\em Conformal invariance under spatial inversion of extreme Reissner-Nordström black holes}, Gen. Relativ. Gravit. {\bf 16} (1984), 8, 789--792.
%
\bibitem{CuGo2015} S.Curry,  A.R. Gover, {\em An Introduction to Conformal Geometry and Tractor Calculus, with a view to Applications in General Relativity}. In T. Daudé, D. Häfner,  J. Nicolas (Eds.), Asymptotic Analysis in General Relativity (2018). (London Mathematical Society Lecture Note Series, pp. 86-170).
%
\bibitem{HFri2004} H. Friedrich, {\em Smoothness at null infinity and the structure of initial data}, in The Einstein equations and the large scale behavior of gravitational fields, p. 121--203, Ed. P. Chrusciel and H. Friedrich, Birkhäuser, Basel, 2004.
%
\bibitem{InglNi} W. Inglese, F. Nicolò, {\em Asymptotic properties of the electromagnetic field in the external Schwarzschild spacetime}, Ann. Henri Poincaré {\bf 1} (2000), 5, 895-944.
%
\bibitem{KuLu2013} H.K. Kunduri, J. Lucietti, {\em Classification of near horizon geometries of extremal black holes}, Living Rev. Relativity {\bf 13} (2013), 8.
%
\bibitem{MaNi2009} L.J. Mason, J.-P. Nicolas, {\em Regularity at space-like and null infinity}, J. Inst. Math. Jussieu {\bf 8} (2009), 1, 179--208.
%
\bibitem{Mo1962} C.S. Morawetz, {\em The decay of solutions of the exterior initial-boundary value problem for the wave equation}, Comm. Pure Appl. Math. {\bf 14} (1961), p. 561--568.
%
\bibitem{Pe65} R. Penrose, {\em Zero rest-mass fields including gravitation~: asymptotic behavior}, Proc. Roy. Soc. {\bf A284} (1965), p. 159--203.
%
\bibitem{PeRi1984} R. Penrose and W. Rindler, Spinors and space-time, Vol. I (1984) and Vol. 2 (1986), Cambridge University Press.
%
\bibitem{Sa61} R. Sachs, {\em Gravitational waves in general relativity VI, the outgoing radiation condition}, Proc. Roy. Soc. A {\bf 264} (1961), 309-338.
%
\bibitem{Sage}
SageMath: a free open-source mathematics software system, \url{https://www.sagemath.org/}
%
\bibitem{LVK2014} C. Lübbe, J.A. Valiente Kroon, {\em On the conformal structure of the extremal Reissner-Nordström spacetime}, Classical and Quantum Gravity, {\bf 31} (2014), 175015, arXiv:1308.1325.
\end{thebibliography}
\end{document}